%% file: IEEE_colored.tex
\documentclass[11pt,onecolumn]{IEEEtran}
\usepackage{rrprefsIEEE}
\begin{document}
\title{Fundamental limit of sample generalized eigenvalue based detection of signals in noise using relatively few signal-bearing and noise-only samples}

\author{Raj Rao Nadakuditi
\thanks{Supported by an Office of Naval Research Special Postdoctoral Award under grant N00014-07-1-0269. Portions of this work were presented at the 2007 Asilomar Conference on Signals and Systems.}
        and Jack W. Silverstein \thanks{Supported by the U.S. Army Research Office under Grant W911NF-05-1-0244.}
}
\thanks{Department of Mechanical Engineering, Massachusetts Institute of Technology, Email: raj@mit.edu, Phone: (857) 891 8303, Fax: (617) 253-4358 }
\thanks{Department of Mathematics, North Carolina State University, Email: jack@unity.ncsu.edu, Phone: (857) 891 8303, Fax: (617) 253-4358 }

\markboth{Sample  generalized eigenvalue based detection}{Nadakuditi and Silverstein}
\maketitle

\input{IEEE_abstract}

\begin{keywords}
signal detection, random matrices, sample covariance matrix, Wishart distribution, multivariate F distribution
\end{keywords}
\begin{center} \bfseries EDICS Category: SSP-DETC Detection; SAM-SDET Source detection
\end{center}

\IEEEpeerreviewmaketitle
\input{spikedcolwish2}

\bibliographystyle{IEEEtran}
\bibliography{randbib}

\end{document}

%% file: IEEE_abstract.tex
\begin{abstract}
The detection problem in statistical signal processing can be succinctly formulated: Given $m$ (possibly) signal bearing, $n$-dimensional signal-plus-noise snapshot vectors (samples) and $N$ statistically independent $n$-dimensional noise-only snapshot vectors, can one reliably infer the presence of a signal? This problem arises in the context of applications as diverse as radar, sonar, wireless communications, bioinformatics, and machine learning and is the critical first step in the subsequent signal parameter estimation phase.

The signal detection problem can  be naturally posed in terms of the sample generalized eigenvalues. The sample generalized eigenvalues correspond to the eigenvalues of the matrix formed by ``whitening'' the signal-plus-noise sample covariance matrix with the noise-only sample covariance matrix. In this article we prove a fundamental asymptotic limit of sample generalized eigenvalue based detection of signals in arbitrarily colored noise when there are relatively few signal bearing and noise-only samples.

Specifically, we show why when the (eigen) signal-to-noise ratio (SNR) is below a critical value, that is a simple function of $n$, $m$ and $N$, then reliable signal detection, in an asymptotic sense, is not possible. If, however, the eigen-SNR is above this critical value then a simple, new random matrix theory based algorithm, which we present here, will reliably detect the signal even at SNR's close to the critical value. Numerical simulations highlight the accuracy of our analytical prediction and permit us to extend our heuristic definition of the \textit{effective number of identifiable signals in colored noise}. We discuss implications of our result for the detection of weak and/or closely spaced signals  in sensor array processing, abrupt change detection in sensor networks, and clustering methodologies in machine learning.
\end{abstract}

%% file: spikedcolwish2.tex
\section{Introduction}
The observation vector, in many signal processing applications, can be modelled as a superposition of a finite number of signals embedded in additive noise. The model order selection problem of inferring the number of signals present is the critical first step in the subsequent signal parameter estimation problem. We consider the class of estimators that determine the model order, \ie, the number of signals, in colored noise from the sample generalized eigenvalues of the signal-plus-noise sample covariance matrix and the noise-only sample covariance matrix pair. The sample generalized eigenvalues \cite{golub96a} precisely correspond to the eigenvalues of the matrix formed by ``whitening'' the signal-plus-noise sample covariance matrix with the noise-only sample covariance matrix (assuming that the number of noise-only samples is greater than the dimensionality of the system so that the noise-only sample covariance matrix is invertible).

Such estimators are used in settings where it is possible to find a portion of the data that contains only noise fields and does not contain any signal information. This is a realistic assumption for many practical applications such as evoked neuromagnetic experiments  \cite{maris03a,sekihara97a,sekihara99a}, geophysical experiments that employ a ``thumper'' or in underwater experiments with a wideband acoustic signal transducer where such a portion can be found in a data portion taken before a stimulus is applied. In applications such as radar or sonar where the signals of interest are narrowband and located in a known frequency band, snapshot vectors collected at a frequency just outside this band can be justified as having the same noise covariance characteristics assuming that we are in the stationary-process-long-observation-time (SPLOT) regime \cite{vantrees02a}.

Our main objective in this paper is to shed new light on this age old problem of detecting signal in noise from finite samples using the sample eigenvalues alone \cite{kailath-wax,zhao86a}. We bring into sharp focus a fundamental statistical limit that  explains precisely when and why, in high-dimensional, sample size limited settings underestimation of the model order is unavoidable. This is in contrast to works in the literature that use simulations, as in \cite{liavas01a}, to highlight the chronically reported symptom of model order estimators underestimating the number of signals without providing insight into whether a fundamental limit of detection is being encountered.

In recent work \cite{raj08a}, we examined this problem in the white noise scenario. The main contribution of this paper is the extension of the underlying idea to the arbitrary (or colored) noise scenario. Analogous to the definition in \cite{raj08a}, we define the \textit{effective number of identifiable signals in colored noise} as the number of the generalized eigenvalues of the population (true) signal-plus-noise covariance matrix and noise-only covariance matrix pair that are greater than a (deterministic) threshold that is a simple function of the number of signal-plus-noise samples, noise-only samples and the dimensionality of the system. Analogous to the white noise case, increasing the dimensionality of the system, by say adding more sensors, raises the detectability threshold so that the effective number of identifiable signals might actually decrease.

An additional contribution of this paper is the development of a simple, new, algorithm for estimating the number of signals based on the recent work of Johnstone \cite{johnstone08a}. Numerical results are used to illustrate the performance of the estimator around the detectability threshold alluded to earlier. Specifically, we observe that if the eigen-SNR of a signal is above a critical value then reliable detection using the new algorithm is possible. Conversely, if the eigen-SNR  is below the critical value then the algorithm, correctly for the reason described earlier, is unable to distinguish the signal from noise.

The paper is organized as follows. We formulate the problem in  Section \ref{sec:lrcf problem formulation} and state the main result in Section \ref{sec:main result}. The effective number of signals is defined in Section \ref{sec:eff signals} along with a discussion on its implications for applications such as array processing, sensor networks and machine learning. A new algorithm for detecting the number of signals is presented in Section \ref{sec:new algorithm}. Concluding remarks are offered in Section \ref{sec:conclusion}. The mathematical proofs of the main result are provided in Section \ref{sec:appendix}.

\section{Problem formulation}
\label{sec:lrcf problem formulation}

We observe $m$ samples (``snapshots'') of possibly signal bearing $n$-dimensional snapshot vectors ${\bf x}_{1}, \ldots, {\bf x}_{m}$ where for each $i$, the snapshot vector has a (real or complex) multivariate normal distribution, \ie, ${\bf x}_{i} \sim \mathcal{N}_{n}(0,{\bf R})$ and the ${\bf x}_{i}$'s are mutually independent.  The snapshot vectors are modelled as
\begin{equation}\label{eq:superposition problem}
{\bf x}_{i} = {\bf A}\,{\bf s}_{i}+{\bf z}_{i}  \qquad \textrm{for } i = 1,\ldots,m,
\end{equation}
where ${\bf z}_{i} \sim \mathcal{N}_{n}(0,\Sigma)$, denotes an $n$-dimensional (real or complex) Gaussian noise vector where the noise covariance $\Sigma$ may be known or unknown, ${\bf s}_{i} \sim \mathcal{N}_{k}({\bf 0},{\bf R}_{s})$ denotes a $k$-dimensional (real or complex) Gaussian signal vector with covariance ${\bf R}_{s}$, and ${\bf A}$ is a $n \times k$  unknown non-random matrix.
Since the signal and noise vectors are independent of each other, the covariance matrix of ${\bf x}_{i}$ can hence be decomposed as
\begin{equation}\label{eq:R model}
{\bf R} =  {\bf \Psi} + \bm{\Sigma}
\end{equation}
where
\begin{equation}
{\bf \Psi} = {\bf A}{\bf R}_{s}{\bf A}' ,
\end{equation}
with $'$ denoting the complex conjugate or real transpose. Assuming that the matrix ${\bf A}$ is of full column rank, \ie, the columns of ${\bf A}$ are linearly independent, and that the covariance matrix of the signals ${\bf R}_{s}$ is nonsingular, it follows that the rank of ${\bf \Psi}$ is $k$. Equivalently, the $n-k$ smallest eigenvalues of ${\bf \Psi}$ are equal to zero.

If the noise covariance matrix $\bm{\Sigma}$ were known apriori and was non-singular, a ``noise whitening'' transformation may be applied to the snapshot vector ${\bf x}_{i}$ to obtain the vector
\begin{equation}
\widetilde{{\bf x}}_{i} = \bm{\Sigma}^{-1/2}{\bf x}_{i},
\end{equation}
which will also be normally distributed with covariance
\begin{equation}\label{eq:Rsigma}
{\bf R}_{\bm{\Sigma}} := \bm{\Sigma}^{-1/2}{\bf R}\bm{\Sigma}^{-1/2} = \bm{\Sigma}^{-1}{\bf \Psi} + {\bf I}.
\end{equation}
Denote the eigenvalues of ${\bf R}_{\bm{\Sigma}}$ by $\lambda_{1}\geq \lambda_{2} \geq \ldots \geq \lambda_{n}$. Recalling the formulation of the generalized eigenvalue problem \cite{golub96a}[Section 8.7], we note that the eigenvalues of ${\bf R}_{\bm{\Sigma}}$ are exactly the generalized eigenvalues of the regular matrix pair $(\widehat{{\bf R}},\widehat{\bm{\Sigma}})$. Then, assuming that the rank of $\bm{\Sigma}^{-1}\bm{\Psi}$ is also $k$, it follows that the smallest $n-k$ eigenvalues of ${\bf R}_{\bm{\Sigma}}$ or, equivalently, the generalized eigenvalues of the matrix pair $({\bf R},\bm{\Sigma})$), are all equal to $1$ so that
\begin{equation}
\lambda_{k+1}= \lambda_{k+2} = \ldots = \lambda_{n} = \lambda = 1,
\end{equation}
while the remaining $k$ eigenvalues ${\bf R}_{\bm{\Sigma}}$ of will be strictly greater than one.

Thus, if the true signal-plus-noise covariance matrix ${\bf R}$ and the noise-only covariance matrix $\bm{\Sigma}$ were known apriori, the number of signals $k$ could be trivially determined from the multiplicity of the eigenvalues of ${\bf R}_{\bm{\Sigma}}$ equalling one.

The problem in practice is that the signal-plus-noise and the noise covariance matrices ${\bf R}$ are unknown so that such a straight-forward algorithm cannot be used. Instead we have an estimate the signal-plus-covariance matrix obtained as
\begin{equation}
\widehat{{\bf R}} = \frac{1}{m} \sum_{i=1}^{m} {\bf x}_{i} {\bf x}_{i}'
\end{equation}
and an estimate of the noise-only sample covariance matrix obtained as
\begin{equation}
\widehat{\bm{\Sigma}} = \frac{1}{N} \sum_{j=1}^{N} {\bf z}_{j} {\bf z}_{j}'
\end{equation}
where ${\bf x}_{i}$ for $i = 1, \ldots, m$  are (possibly) signal-bearing snapshots and ${\bf z}_{j}$ for $j = 1, \ldots, N$ are independent noise-only snapshots. We assume here that the number of noise-only snapshots exceeds the dimensionality of the system, \ie, $N>n+1$, so that the noise-only sample covariance matrix $\widehat{\bm{\Sigma}}$, which has the Wishart distribution \cite{wishart28a}, is non-singular and hence invertible with probability 1  \cite[Chapter 3, pp. 97]{muirhead82a},\cite[Chapter 7.7, pp. 272-276]{anderson03a}.  Following (\ref{eq:Rsigma}), we then form the matrix
\begin{equation}\label{eq:RhatSigmahat}
\widehat{{\bf R}}_{\widehat{\Sigma}} = \widehat{\bm{\Sigma}}^{-1}\widehat{{\bf R}},
\end{equation}
and compute its eigen-decomposition to obtain the eigenvalues of $\widehat{{\bf R}}_{\widehat{\Sigma}}$, which we denote by $\hat{\lambda_{1}} \geq \hat{\lambda}_{2} \geq \ldots \geq \hat{\lambda}_{n}$. We note, once again, that the eigenvalues of $\widehat{{\bf R}}_{\widehat{\Sigma}}$ are simply the generalized eigenvalues of the regular matrix pair $(\widehat{{\bf R}},\widehat{\bm{\Sigma}})$.  Note that whenever $N < n$, the signal-plus-noise sample covariance matrix ${\bf R}$ will be singular so that the $n-N$ generalized eigenvalues will equal zero, \ie, $\hat{\lambda}_{N+1} = \hat{\lambda}_{N+2} = \ldots = \hat{\lambda}_{n} = 0$. Figure \ref{fig:eig blurring} illustrates why the blurring of the sample eigenvalues relative to the population eigenvalues makes the problem more challenging.

In this paper, we are interested in the class of algorithms that infer the number of signals buried in arbitrary noise from the eigenvalues of $\widehat{{\bf R}}_{\widehat{\bm{\Sigma}}}$ or $\widehat{{\bf R}}_{\bm{\Sigma}}$ alone.  Such algorithms are widely used in practice and arise naturally from classical multivariate statistical theory \cite{johnstone08a} where the matrix $\widehat{{\bf R}}_{\widehat{\Sigma}}$ is referred to as the multivariate F matrix  \cite{muirhead82a,silverstein85a}. The information theoretical approach to model order estimation, first introduced by Wax and Kailath \cite{kailath-wax}, was extended to the colored noise setting by Zhao et al in \cite{zhao86b} who prove consistency of their estimator in the large sample size regime; their analysis does not yield any insight into the finite sample setting.

Consequently, research has focussed on developing sophisticated techniques for improving performance of eigenvalue based methods in the finite sample setting. Zhu et al \cite{zhu91a} improve the performance of their eigenvalue estimator by assuming a model for the noise covariance matrix. Stoica and Cedervall \cite{stoica97a}  improve the performance of their estimator in two reasonable settings: one, where it is reasonable to assume that the noise covariance matrix is block diagonal or banded and two, where the temporal correlation of the noise has a shorter length than the signals. Other techniques in the literature exploit other characteristics of the signal or noise to effectively reduce the dimensionality of the signal subspace and improve model order estimation given finite samples. See for example  \cite{xu94a,larocque02a} and the references in \cite{raj08a}.

Informally speaking, it is evident that performance of such model order estimation algorithms is coupled to the ``quality'' of the estimated signal-plus-noise and noise-only  covariance matrices which in turn are dependent on the number of snapshots used to estimate them, respectively. Researchers applying these techniques have noted the absence of a mathematically rigorous, general purpose formula in the literature for predicting the minimum number of samples needed to obtain ``good enough'' detection accuracy (see, for example \cite{sekihara97a}[pp. 846]. A larger, more fundamental question that has remained unanswered, till now, is whether there is a statistical limit being encountered.

We tackle this problem head on in this paper by employing sophisticated techniques from random matrix theory in \cite{silverstein:book}. We show that in an asymptotic sense, to be made precise later, that only the ``signal'' eigenvalues of ${\bf R}_{\Sigma}$ that are above a deterministic threshold can be reliably distinguished from the ``noise'' eigenvalues. The threshold is a simple, deterministic function  of the the dimensionality of the system, the number of noise-only and signal-plus-noise snapshots, and the noise and signal-plus noise covariance, and described explicitly next. Note the applicability of the results to the situation when the signal-plus-noise covariance matrix is singular.

\begin{figure}[t]
\centering
\includegraphics[height=3.5in]{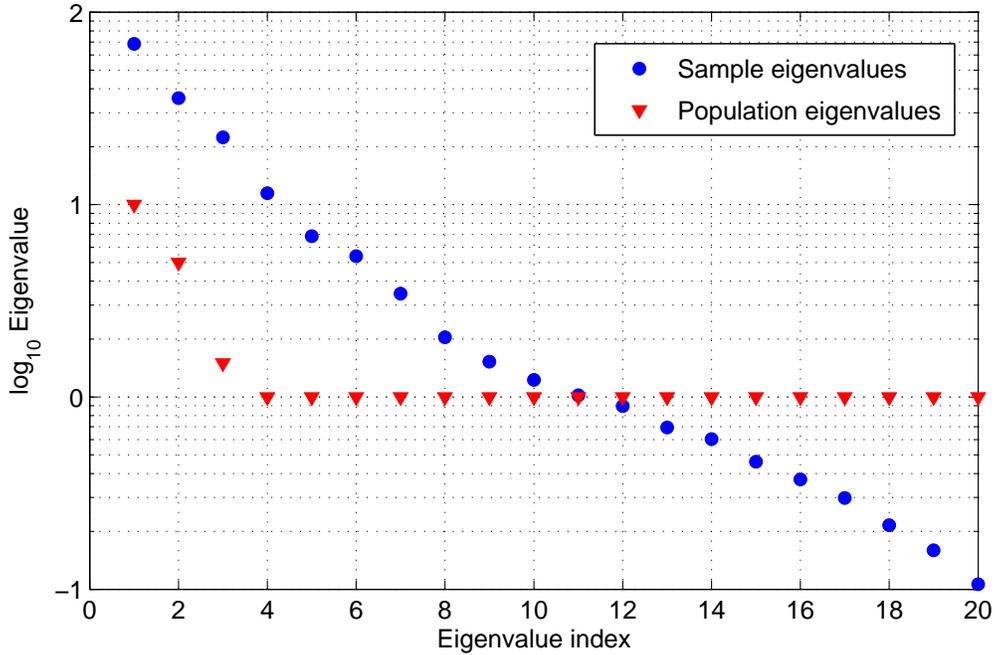}
\caption{The dimension of the ``noise'' subspace is equal to the multiplicity of the population eigenvalue equal to one. When the population eigenvalues are known, then detecting the number of signals becomes trivial. However, estimating the number of signals from the sample generalized eigen-spectrum is considerably more challenging because of the finite sample effects. Specifically, the finite number of noise-only and signal-plus-noise samples induces a blurring in the sample eigenspectrum relative to the population eigenspectrum makes discrimination of the ``signal'' from the ``noise'' challenging. The figure shows one random instance generated for a $n=20$ dimensional system with $N = 25$ noise-only samples and $m = 40$ signal-plus-noise bearing samples.}
\label{fig:eig blurring}
\end{figure}

\section{Main result}\label{sec:main result}
For a Hermitian matrix ${\bf A}$ with $n$ real eigenvalues (counted with multiplicity), the empirical distribution function (e.d.f.) is defined as
\begin{equation}\label{eq:edf definition}
F^{A}(x) = \frac{\textrm{Number of eigenvalues of } {\bf A} \leq x}{n}.
\end{equation}
Of particular interest is the convergence of the e.d.f. of $\widehat{{\bf R}}_{\widehat{\Sigma}}$ in the signal-free case, which is described next.

\begin{theorem}\label{prop:mandp density}
Let $\widehat{{\bf R}}_{\widehat{\Sigma}}$ denote the matrix in (\ref{eq:RhatSigmahat}) formed from $m$ (complex Gaussian) noise-only snapshots and $N$ independent noise-only (complex Gaussian) snapshots. Then the e.d.f. $F^{\widehat{{\bf R}}_{\widehat{\Sigma}}}(x) \to F^{R_{\Sigma}}(x)$ almost surely for every $x$, as $m,n(m) \to \infty$, $m,N(m) \to \infty$ and $c_{m} = n/m \to c >0$ and $c^{1}_{N}=n/N \to c_1 <1$ where

\begin{equation}\label{eq:mandp density}
dF(x) = \max\left(0,\left(1-\frac{1}{c}\right)\right)\delta(x)+\frac{(1-c_1)\sqrt{(x-b_1)(b_2-x)}}{2\pi x(xc_1+c)}\mathbb{I}_{[b_1,b_2]}(x)\, dx,
\end{equation}
where
\begin{equation}
b_1=\left(\frac{1-\sqrt{1-(1-c)(1-c_1)}}{1-c_1}\right)^2,\quad
b_2=\left(\frac{1+\sqrt{1-(1-c)(1-c_1)}}{1-c_1}\right)^2,
\end{equation}
$\mathbb{I}_{[b_1,b_2]}(x) = 1$ when $b_1 \leq x \leq b_2$ and zero otherwise, and $\delta(x)$ is the Dirac delta function.
\end{theorem}
\begin{proof}
This result was proved in \cite{silverstein85a}. When $c_{1} \to 0$ we recover the famous Mar\v{c}enko-Pastur density \cite{marcenko67a}.
\end{proof}

The following result exposes when the ``signal'' eigenvalues are asymptotically distinguishable from the ``noise'' eigenvalues.

\begin{theorem}\label{prop:spiked convergence}
Let $\widehat{{\bf R}}_{\widehat{\Sigma}}$ denote the matrix in (\ref{eq:RhatSigmahat}) formed from $m$ (real or complex Gaussian) signal-plus-noise snapshots and $N$ independent (real or complex Gaussian) noise-only snapshots. Denote the eigenvalues of ${{\bf R}}_{{\Sigma}}$ by $\lambda_{1} \geq \lambda_{2} > \ldots \geq \lambda_{k} > \lambda_{k+1} = \ldots \lambda_{n} = 1$. Let $l_{j}$ denote the $j$-th largest eigenvalue of $\widehat{{\bf R}}_{\widehat{\Sigma}}$. Then as $n,m(n) \to \infty$, $n,N(n) \to \infty$ and $c_{m} = n/m \to c >0$ and $c^{1}_{N}=n/N \to c_1 <1$ we have

\begin{figure*}[h]
\begin{equation*}
l_{j} \to
\begin{cases}
\lambda_{j} \left( 1-c-c {\dfrac {-c_{1}\,\lambda_{j}-\lambda_{j}+1+\sqrt {{c_{1}}^{2}{\lambda_{j}}^{2}-2\,c_{1}\,{\lambda_{j}}^{2}-2\,c_{1}\,\lambda_{j}+{\lambda_{j}}^{2}-2\,\lambda_{j}+1}}{2c_{1}\,{\lambda_{j}}}}\right),  &\lambda_{j} > \tau(c,c_{1}) \\
& \\
{\dfrac {-\,c_{1}\,c+\,c+1+\,c_{1} + 2\,\sqrt {c+c_{1}-c_{1}c}}{{c_{1}}^{2}+1-2\,c_{1}}},  &\lambda_{j} \leq \tau(c,c_{1}) \\
\end{cases}
\hrulefill
\end{equation*}
\end{figure*}
for $j =1, \ldots, k$ and the convergence is almost surely and the threshold ${\rm T}(c,c_{1})$ is given by

\begin{equation}\label{eq:thresholdSNR}
{\rm T}(c,c_1) = \dfrac{1+\tau-\tau c_1+\sqrt{(1+\tau-\tau c_1)^2-4\tau}}{2},
\end{equation}
where
\begin{equation}
\tau=\dfrac{(1+c_1)\alpha+\sqrt{\alpha}\sqrt{4\alpha-c_1+(1-c_1)^2c^2}}
{(1-c_1)^2\alpha}=\frac{(1+c_1)\alpha+\sqrt{\alpha}(2c_1+c(1-c_1))}
{(1-c_1)^2\alpha}
\end{equation}
and $\alpha = c+c_1-c_1c$.
\end{theorem}
\begin{proof}
The result follows from Theorem \ref{th:jack 4}. The threshold $\Tau(c,c_{1})$ is obtained by solving the inequality
\[
t' > \tau
\]
where for $j = 1, \ldots, k$, $t'$, from \cite{BaikBP04,BaikS06,paul07a,raj08a}, is given by
\[
t' = \dfrac{1}{\lambda_{j} \left(1+ \dfrac{c_{1}}{\lambda_{j}-1}\right)}
\]
and $\tau$ is given by (\ref{eq:jack 9}).

Note that when $c_{1} \to 0$, $\Tau(c,c_1) \to (1+\sqrt{c})$ so that we recover the results of Baik and Silverstein \cite{BaikS06}.
\end{proof}

\subsection{Effective number of identifiable signals}\label{sec:eff signals}

Theorem \ref{prop:spiked convergence} brings into sharp focus the reason why, in the large-system-relatively-large-sample-size limit, model order underestimation is sometimes unavoidable. This motivates our heuristic definition of the \textit{effective number of identifiable signals} below:
\begin{equation}\label{eq:eff signals}
k_{eff}({\bf R},\bm{\Sigma}) = \textrm{\# Eigs. of } \bm{\Sigma}^{-1}{\bf R} > {\rm T}(c,c_1) \approx {\rm T}\left(\dfrac{n}{m},\dfrac{n}{N}\right).
\end{equation}

If we denote the eigenvalues of ${\bf R}_{\bm{\Sigma}} \equiv  \bm{\Sigma}^{-1}{\bf R}$ by $\lambda_{1} \geq \lambda_{2} > \ldots \geq \lambda_{k} > \lambda_{k+1} = \ldots \lambda_{n} = 1$ then we define the eigen-SNR of the $j$-th signal as $\lambda_{j} -1$ then (\ref{eq:eff signals}) essentially states that signals with eigen-SNR's smaller than $\Tau(n/m,n/N)$ will be asymptotically undetectable.

Figure \ref{fig:eigensnr} shows the eigen-SNR threshold $\Tau(c,c_{1})-1$ needed for reliable detection for different values as a function of $c$ for different values of $1/c_{1}$. Such an analytical prediction was not possible before the results presented in this paper. Note the fundamental limit of detection in the situation when the noise-only covariance matrix is known apriori (solid line) and increase in the threshold eigen-SNR needed as the number of snapshots available to estimate the noise-only covariance matrix decreases.

\begin{figure}[t]
\centering
\includegraphics[width=6.5in]{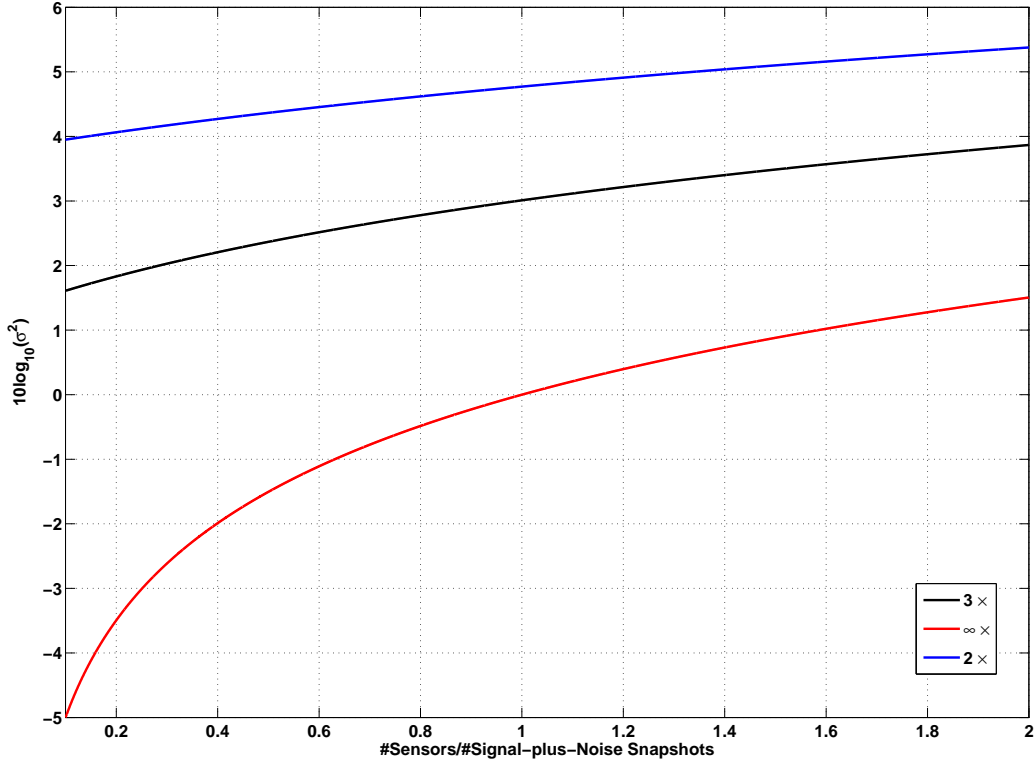}
\caption{Plot of the minimum (generalized) Eigen-SNR required (equal to $\Tau(c,c_{1})-1$ where $\Tau(c,c_1)$ is given by (\ref{eq:thresholdSNR})) to be able to asymptotically discriminate between the ``signal'' and ``noise'' eigenvalue of the matrix $\widehat{{\bf R}}_{\widehat{\Sigma}}$ constructed as in (\ref{eq:RhatSigmahat}) as a function of the ratio of the number of sensors to snapshots for different values of $1/c_{1}$ where $c_{1} \approx $ Number of sensors/Number of noise-only snapshots.  The gap between the upper two lines and the bottom most line represents the SNR loss due to noise covariance matrix estimation.}
\label{fig:eigensnr}
\end{figure}

\subsection{Implications for array processing}

Suppose there are two uncorrelated (hence, independent) signals so that ${\bf R}_{s} = \textrm{diag}(\sigma_{{\rm S}1}^{2},\sigma_{{\rm S}2}^{2})$. In (\ref{eq:superposition problem}) let ${\bf A} = [{\bf v}_{1} {\bf v}_{2}]$. In a sensor array processing application, we think of ${\bf v}_{1} \equiv {\bf v}(\theta_{1})$ and ${\bf v}_{2} \equiv {\bf v}_{2}(\theta_{2})$ as encoding the array manifold vectors for a source and an interferer with powers $\sigma_{{\rm S}1}^{2}$ and $\sigma_{{\rm S}2}^{2}$, located at $\theta_{1}$ and $\theta_{2}$, respectively. The signal-plus-noise covariance matrix is given by
\begin{equation}
{\bf R} = \sigma_{{\rm S}1}^{2} {\bf v}_{1}{\bf v}_{1}'+ \sigma_{{\rm S}2}^{2}  {\bf v}_{2}{\bf v}_{2}' + \bm{\Sigma}
\end{equation}
where $\bm{\Sigma}$ is the noise-only covariance matrix. The matrix ${\bf R}_{\Sigma}$ defined in (\ref{eq:Rsigma}) can be decomposed as
\[
{\bf R}_{\Sigma} = \bm{\Sigma}^{-1}{\bf R} =  \sigma_{{\rm S}1}^{2} \bm{\Sigma}^{-1}{\bf v}_{1}{\bf v}_{1}'+ \bm{\Sigma}^{-1}\sigma_{{\rm S}2}^{2}  {\bf v}_{2}{\bf v}_{2}' + {\bf I}
\]
so we that we can readily note that ${\bf R}_{\Sigma}$ has the $n-2$ smallest eigenvalues $\lambda_{3} = \ldots = \lambda_{n} = 1$ and the two largest eigenvalues
\begin{subequations}\label{eq:ev 2 sources}
\begin{equation}
\lambda_{1} =
1+ \dfrac{\left(\sigma_{{\rm S}1}^{2} \parallel \! {\bf u}_{1} \!\parallel^{2}+\sigma_{{\rm S}2}^{2} \parallel \! {\bf u}_{2} \!\parallel^{2}\right)}{2} + \dfrac{
\sqrt{\left(\sigma_{{\rm S}1}^{2} \parallel \! {\bf u}_{1} \!\parallel^{2}-\sigma_{{\rm S}2}^{2} \parallel \! {\bf u}_{2} \!\parallel^{2}\right)^{2}+4\sigma_{{\rm S}1}^{2}\sigma_{{\rm S}2}^{2} |\langle {\bf u}_{1}, {\bf u}_{2} \rangle| ^{2}}}{2}
\end{equation}
\begin  {equation}
\lambda_{2} =  1+ \dfrac{\left(\sigma_{{\rm S}1}^{2} \parallel \! {\bf u}_{1} \!\parallel^{2}+\sigma_{{\rm S}2}^{2} \parallel \! {\bf u}_{2} \!\parallel^{2}\right)}{2} -\dfrac{
\sqrt{\left(\sigma_{{\rm S}1}^{2} \parallel \! {\bf u}_{1} \!\parallel^{2}-\sigma_{{\rm S}2}^{2} \parallel \! {\bf u}_{2} \!\parallel^{2}\right)^{2}+4\sigma_{{\rm S}1}^{2}\sigma_{{\rm S}2}^{2} |\langle {\bf u}_{1}, {\bf u}_{2} \rangle| ^{2}}}{2}
\end{equation}
\end{subequations}
respectively, where ${\bf u}_{1} := \bm{\Sigma}^{-1/2}{\bf v}_{1}$ and ${\bf u}_{2} := \bm{\Sigma}^{-1/2}{\bf v}_{2}$ . Applying the result in Theorem \ref{prop:spiked convergence} allows us to express the effective number of signals as
\begin{equation}\label{eq:array tradeoff}
k_{{\rm eff}} =
\begin{cases}
2 &\qquad\textrm{if    } \phantom{~~~~} \Tau\left(\frac{n}{m},\frac{n}{N}\right) < \lambda_{2}\\
  &\\
1 &\qquad\textrm{if    } \phantom{~~~~}\lambda_{2} \leq \Tau\left(\frac{n}{m},\frac{n}{N}\right) < \lambda_{1}\\
   &\\
0 &\qquad\textrm{if    } \phantom{~~~~}\lambda_{1} \leq \Tau\left(\frac{n}{m},\frac{n}{N}\right). \\
\end{cases}
\end{equation}

Equation (\ref{eq:array tradeoff}) captures the tradeoff between the identifiability of two closely spaced signals, the dimensionality of the system, the number of available snapshots and the cosine of the angle between the vectors ${\bf v}_{1}$ and ${\bf v}_{2}$. Note that since the effective number of signals depends on the structure of the theoretical signal and noise covariance matrices (via the eigenvalues of ${\bf R}_{\Sigma}$), different assumed noise covariance structures (AR(1) versus white noise, for example) will impact the signal level SNR needed for reliable detection in different ways.

\subsection{Other applications}
There is interest in detecting abrupt change in a system based on stochastic observations of the system using a network of sensors. When the observations made at various sensors can be modeled as Gauss-Markov random field (GMRF), as in \cite{anandkumar09a,sung06a}, then the conditional independence property of GMRF's \cite{rue2005gmr} is a useful assumption. The assumption states that conditioned on a particular hypothesis, the observations at sensors are independent. This assumption results in the precision matrix, \ie, the inverse of the covariance matrix, having a sparse structure with many entries identically equal to zero.

Our results might be used to provide insight into the types of systemic changes, reflected in the structure of the signal-plus-noise covariance matrix, that are undetectable using sample generalized eigenvalue based estimators. Specifically, the fact that the inverse of the noise-only covariance matrix will have a sparse structure means that one can experiment with different (assumed) conditional independence structures and determine how ``abrupt'' the system change would have to be in order to be reliably detected using finite samples.

Spectral methods are popular in machine learning applications  such as unsupervised learning, image segmentation, and information retrieval \cite{vempala2007sal}. Generalized eigenvalue based techniques for clustering have been investigated in \cite{mangasarian2006mps,guarracino2007cmb}. Our results might provide insight when spectral clustering algorithms are likely to fail. In particular, we note that the results of Theorem \ref{prop:spiked convergence} hold even in  situations where the data is not Gaussian (see Theorem \ref{th:jack 4}) as is commonly assumed in machine learning applications.

\section{An algorithm for reliable detection of signals in noise}\label{sec:new algorithm}

In \cite{johnstone08a}, Johnstone proves that in the signal-free case, the distribution of the largest eigenvalue of $\widehat{{\bf R}}_{\widehat{\bm{\Sigma}}}$, on appropriate centering and scaling, can be approximated to order $O(n^{-2/3})$ by the Tracy-Widom law \cite{TracyW94,TracyW96, johnstone07a}. In the setting where there are signals present, we expect, after appropriate centering and scaling, the distribution of the signal eigenvalues of $\widehat{{\bf R}}_{\widehat{\bm{\Sigma}}}$ above the detectability threshold will obey a Gaussian law whereas those below the detectability threshold will obey the Tracy-Widom law as in the signal-free case. An analogous results for the signal bearing eigenvalues of $\widehat{{\bf R}}_{\bm{\Sigma}}$ was proved by Baik et al \cite{BaikBP04} and El Karoui \cite{ElKaroui07a}. Numerical investigations for (see Figure \ref{fig:cdf comparison}) corroborate the accuracy of our asymptotic predictions and form the basis of Algorithm 1 presented below for estimating the number of signals at (asymptotic) significance level $\alpha$. Theoretical support for this observation remains incomplete.

\begin{table*}[h]
\centering
\begin{tabular}{l}
\hline
{\sf Algorithm 1} \\
\hline
Input: Eigenvalues $\widehat{\lambda}_{j}$ for $j = 1, \ldots, n$ of $\widehat{{\bf R}}_{\widehat{\bm{\Sigma}}}$\\
1. Initialization: Set significance level $\alpha \in (0,1)$ \\
2. Compute $\tau_{\alpha} := TW_{\{\mathbb{R},\mathbb{C}\}}^{-1}(1-\alpha)$ from Table \ref{tab:tw quantiles}\\
3. Set \textcolor{red}{k} = 0 \\
4. Compute $\mu_{\{\mathbb{R},\mathbb{C}\}}[n-\textcolor{red}{k},m]$ and $\sigma_{\{\mathbb{R},\mathbb{C}\}}[n-\textcolor{red}{k},m]$ from Table \ref{tab:alg jac}\\
5. Is $\dfrac{\log m \hat{\lambda}_{\textcolor{red}{k+1}}/N - \mu_{\{\mathbb{R},\mathbb{C}\}}[n-\textcolor{red}{k},m-\textcolor{red}{k},N]}{\sigma_{\{\mathbb{R},\mathbb{C}\}}[n-\textcolor{red}{k},m-\textcolor{red}{k},N]} \geq \tau_{\alpha}$?\\
6. If yes, then go to step 9\\
7. Otherwise, increment $\textcolor{red}{k}$. \\
8. If $\textcolor{red}{k} < \min(n,m)$, go to step 3. Else go to step 9.\\
9. Return $\widehat{k} = \textcolor{red}{k}$ \\
\hline
\end{tabular}
\end{table*}

Figure \ref{fig:threshold experiment} illustrates the accuracy of the predicted statistical limit and the ability of the proposed algorithm to reliably detect the presence of the signal at this limit.

\begin{figure}[t]
\centering
\subfigure[Here $\sigma^{2} = 0.5$, so that $\lambda_{1} = 1+\sigma^{2}= 1.5 < \Tau({320}{160},{320}{960}) = 3.4365$.]{
\includegraphics[width=5.5in]{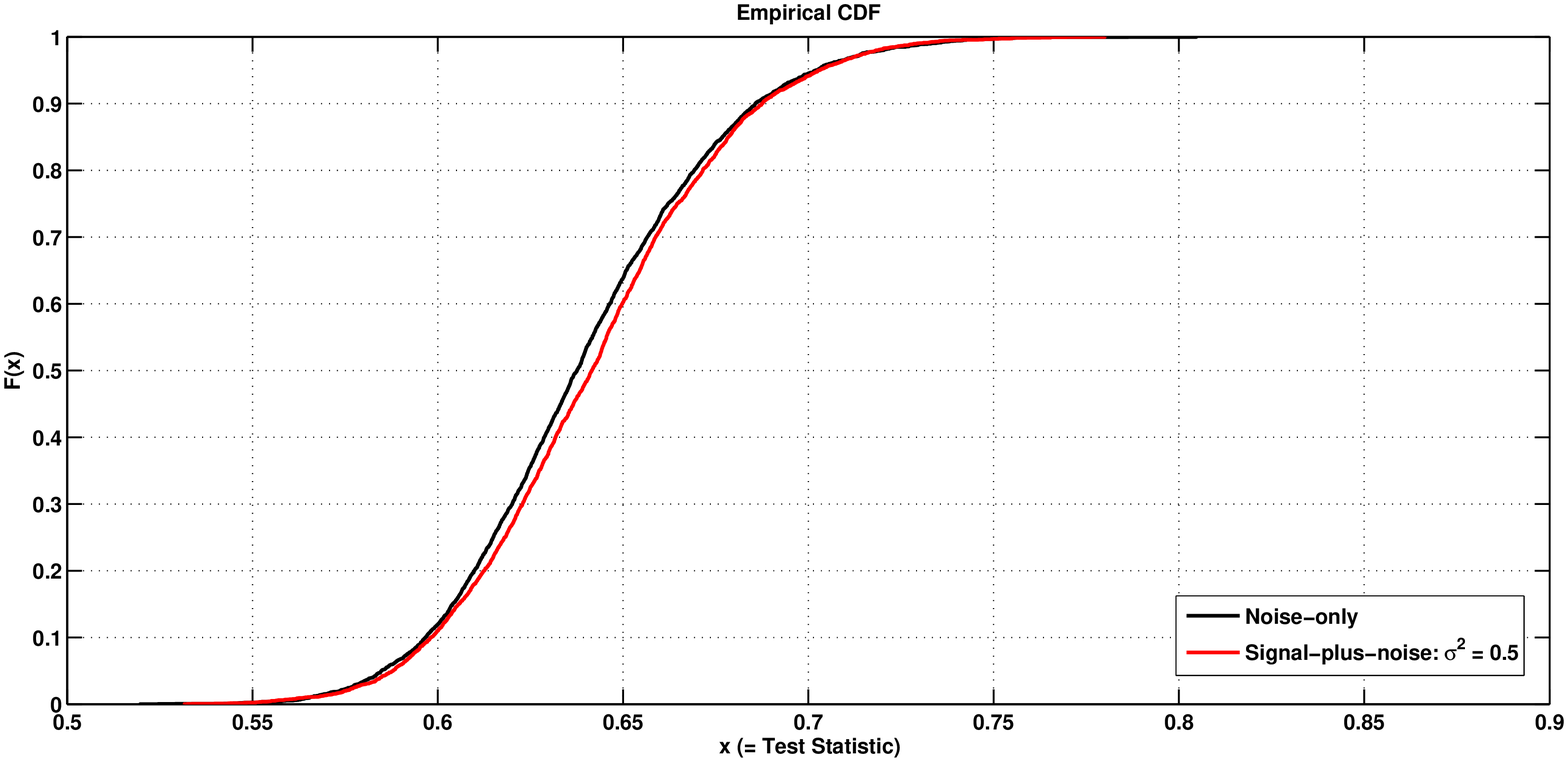}
\label{fig:cdf1}
}\\[0.75in]
\subfigure[Here $\sigma^{2} = 5$, so that $\lambda_{1} = 1+\sigma^{2} = 6 > \Tau(\frac{320}{160},\frac{320}{960}) = 3.4365$]{
\includegraphics[width=5.5in]{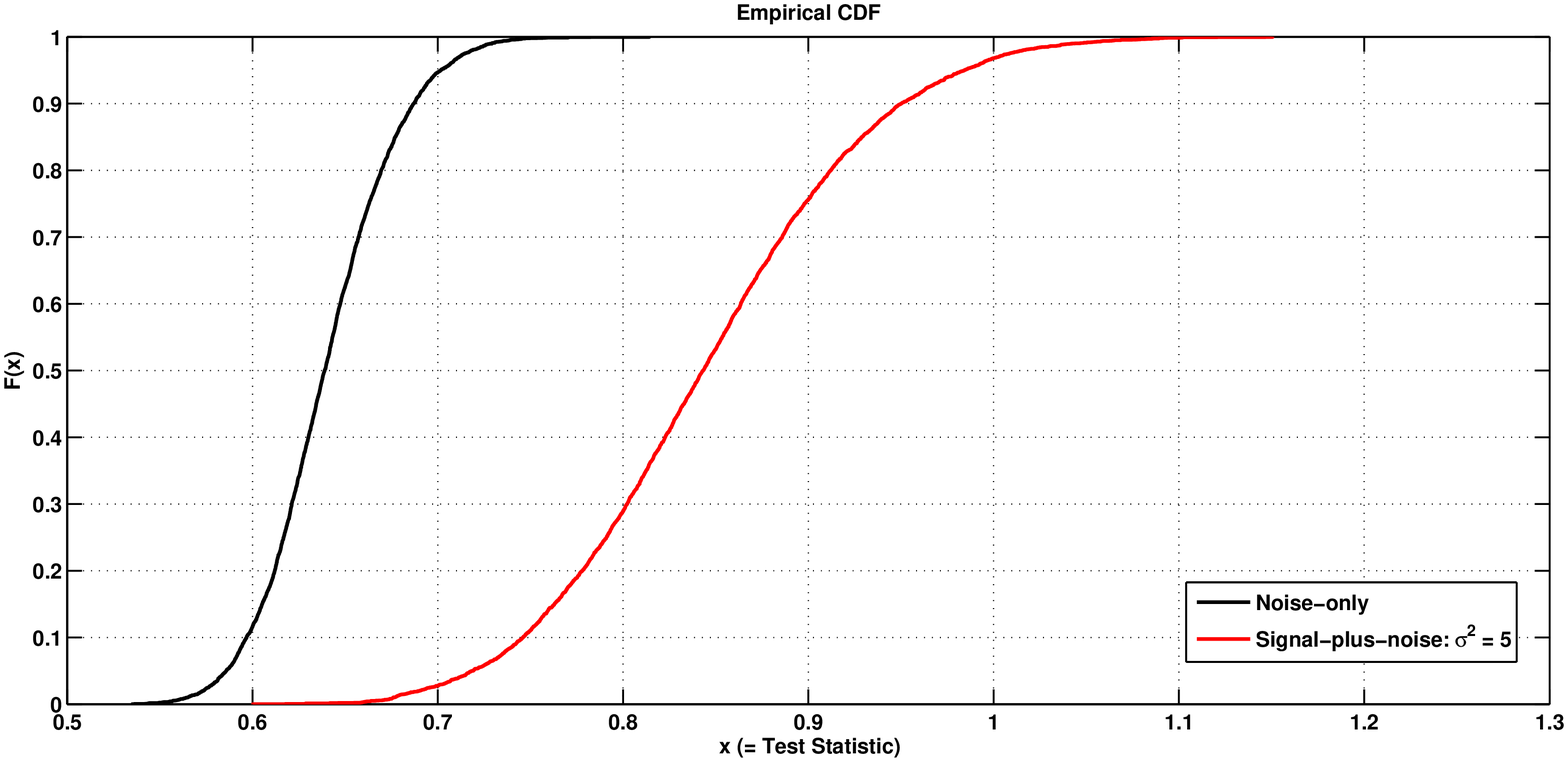}
\label{fig:cdf2}
}
\caption{In (a), for the setting described in Theorem \ref{prop:spiked convergence} we set $n = 320$, $m = 160$, $N = 960$, $\sigma^{2} =0.5$, and w.l.o.g. $\bm{\Sigma} = {\bf I}$, ${\bf R} = \textrm{diag}(\lambda_{1} = 1+ \sigma^{2},1,\ldots, 1)$ and compare the  the empirical cdf of the largest eigenvalue of $\widehat{{\bf R}}_{\widehat{\bm{\Sigma}}}$ with the largest eigenvalue of $\widehat{{\bf R}}_{\widehat{\bm{\Sigma}}}$ with ${\bf R} = {\bf I}$, \ie, in the noise-only case,  over $1000$ Monte-Carlo trials. In (b), we plot the empirical cdf but now with $\sigma^{2} = 5$.}
\label{fig:cdf comparison}
\end{figure}

\begin{sidewaystable}
\subtable[Algorithm 1]{
\label{tab:alg jac}
\includegraphics[width = 8.5in]{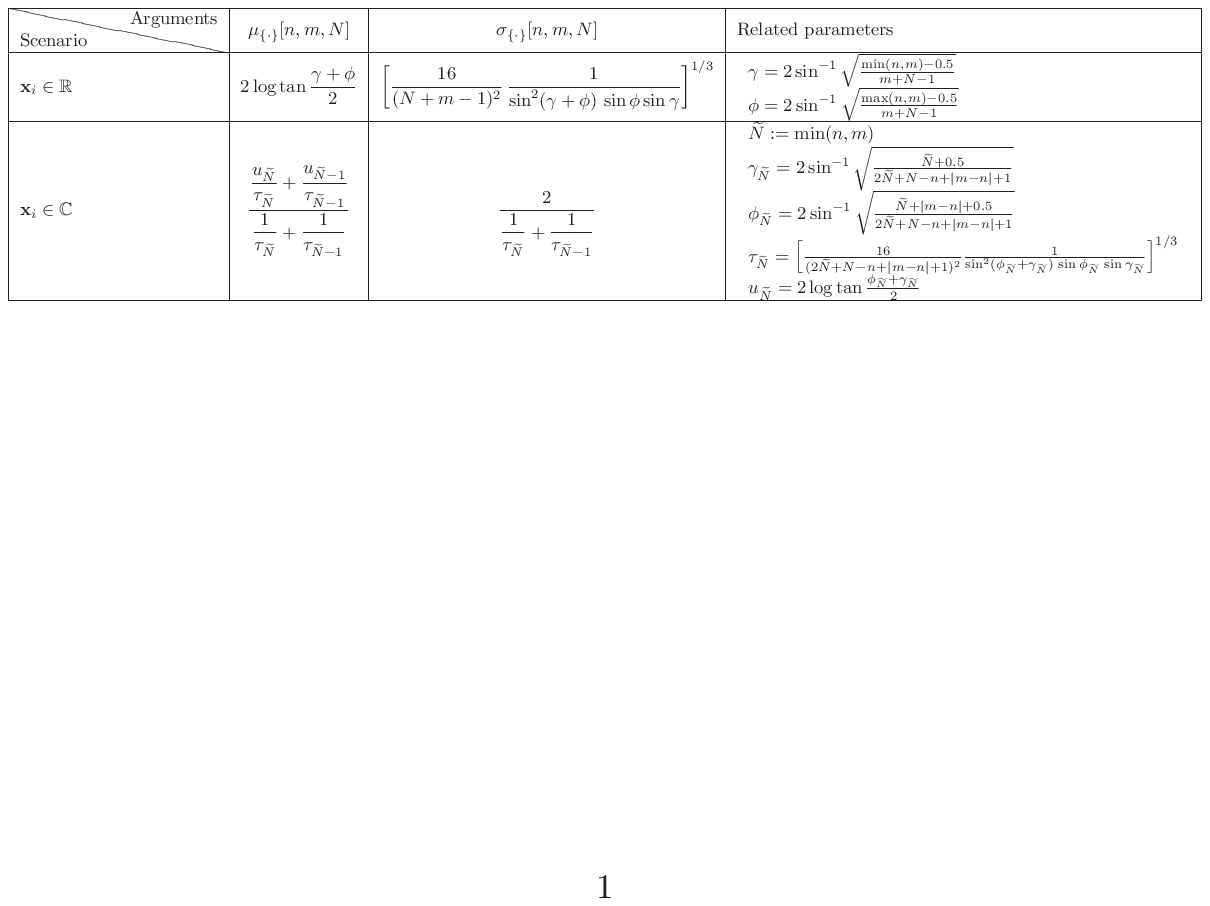}
}\\[0.1in]
\subtable[Algorithm 2]{\label{tab:alg wish}
\includegraphics[width = 8.55in]{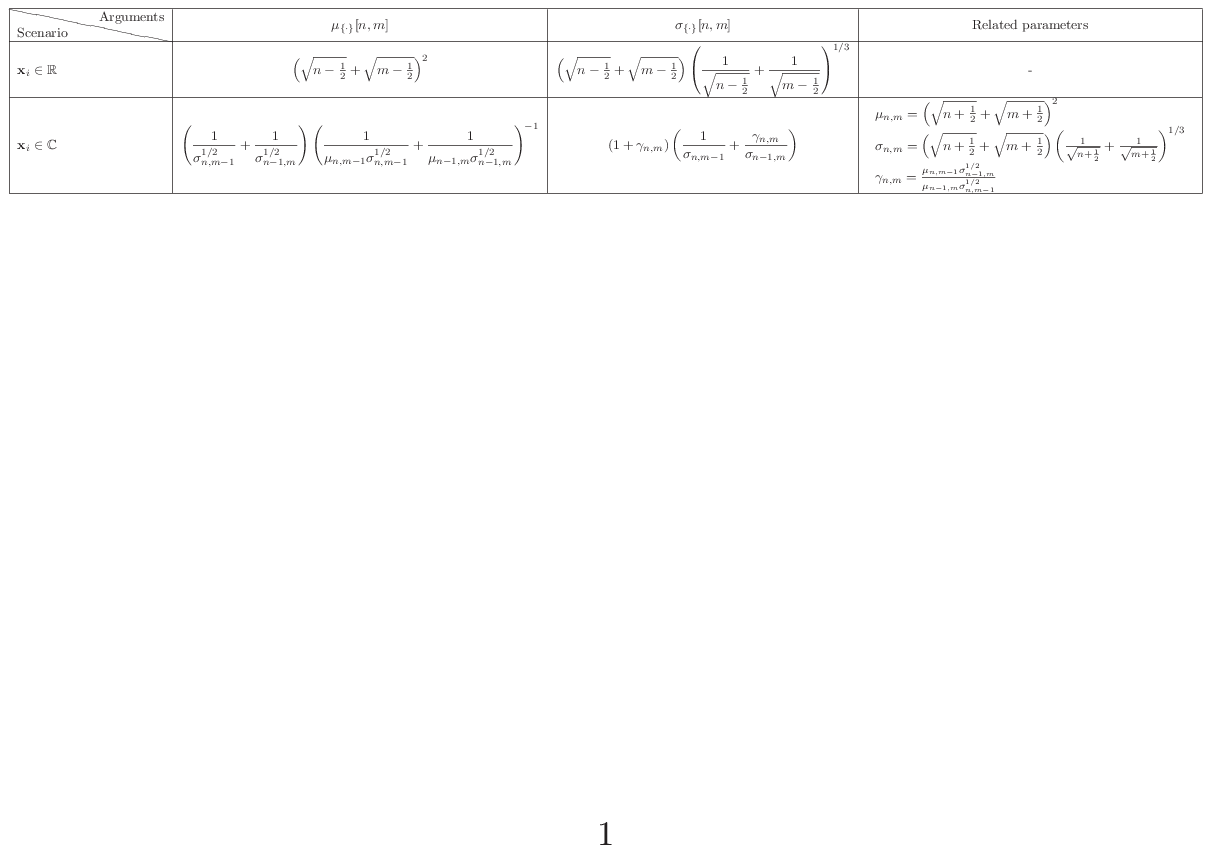}
}
\caption{Parameters for signal detection algorithms.}
\label{tab:alg parameters}
\end{sidewaystable}

\begin{table}
\centering
\begin{tabular}{|c|c|c|c|}
\hline
$\alpha$ &$1-\alpha$ & $TW^{-1}_{\mathbb{R}}(1-\alpha)$ &$TW^{-1}_{\mathbb{C}}(1-\alpha)$ \\
\hline
0.990000  &0.010000  &-3.89543267306429  &-3.72444594640057\\
0.950000  &0.050000  &-3.18037997693774  &-3.19416673215810\\
0.900000  &0.100000  &-2.78242790569530  &-2.90135093847591\\
0.700000  &0.300000  &-1.91037974619926  &-2.26618203984916\\
0.500000  &0.500000  &-1.26857461658107  &-1.80491240893658\\
0.300000  &0.700000  &-0.59228719101613  &-1.32485955606020\\
0.100000  &0.900000  &0.45014328905825   &-0.59685129711735\\
0.050000  &0.950000  &0.97931605346955   &-0.23247446976400\\
0.010000  &0.990000  &2.02344928138015   &0.47763604739084\\
0.001000  &0.999000  &3.27219605900193   &1.31441948008634\\
0.000100  &0.999900  &4.35942034391365   &2.03469175457082\\
0.000010  &0.999990  &5.34429594047426   &2.68220732168978\\
0.000001 &0.999999  &6.25635442969338   &3.27858828203370\\
\hline
\end{tabular}
\caption{The third and fourth column show the percentiles of the Tracy-Widom real  and complex distribution respectively corresponding to fractions in the second column. The percentiles were computed in \matlab  using software provided by Folkmar Bornemann for the efficient evaluation of  the real and complex Tracy-Widom distribution functions $F^{TW}_{\{\mathbb{R},\mathbb{C}\}}(x)$. The percentiles are computed using the \texttt{fzero} command in \matlab. The accuracy of the computed percentiles is about $\pm 5 \times 10^{-15}$ in absolute error terms.}
\label{tab:tw quantiles}
\end{table}

In the special setting where the noise covariance matrix is known apriori, the results of Baik et al \cite{BaikBP04}, El Karoui \cite{ElKaroui07a} and Ma \cite{ma08a} form the basis of Algorithm 2 presented below for estimating the number of signals at (asymptotic) significance level $\alpha$.

\begin{table*}[hp]
\centering
\begin{tabular}{l}
\hline
{\sf Algorithm 2}\\
\hline
Input: Eigenvalues $\widehat{\lambda}_{j}$ for $j = 1, \ldots, n$ of $\widehat{{\bf R}}_{\bm{\Sigma}}$\\
1. Initialization: Set significance level $\alpha \in (0,1)$ \\
2. Compute $\tau_{\alpha} := TW_{\{\mathbb{R},\mathbb{C}\}}^{-1}(1-\alpha)$ from Table \ref{tab:tw quantiles}\\
3. Set \textcolor{red}{k} = 0 \\
4. Compute $\mu_{\{\mathbb{R},\mathbb{C}\}}[n-\textcolor{red}{k},m]$ and $\sigma_{\{\mathbb{R},\mathbb{C}\}}[n-\textcolor{red}{k},m]$ from Table \ref{tab:alg wish}\\
5. Is $\dfrac{m \hat{\lambda}_{\textcolor{red}{k+1}} - \mu_{\{\mathbb{R},\mathbb{C}\}}[n-\textcolor{red}{k},m]}{\sigma_{\{\mathbb{R},\mathbb{C}\}}[n-\textcolor{red}{k},m]} \geq \tau_{\alpha}$?\\
6. If yes, then go to step 9\\
7. Otherwise, increment $\textcolor{red}{k}$. \\
8. If $\textcolor{red}{k} < \min(n,m)$, go to step 3. Else go to step 9.\\
9. Return $\widehat{k} = \textcolor{red}{k}$ \\
\hline
\end{tabular}
\end{table*}

\begin{figure}[t]
\centering
\includegraphics[width=6.5in]{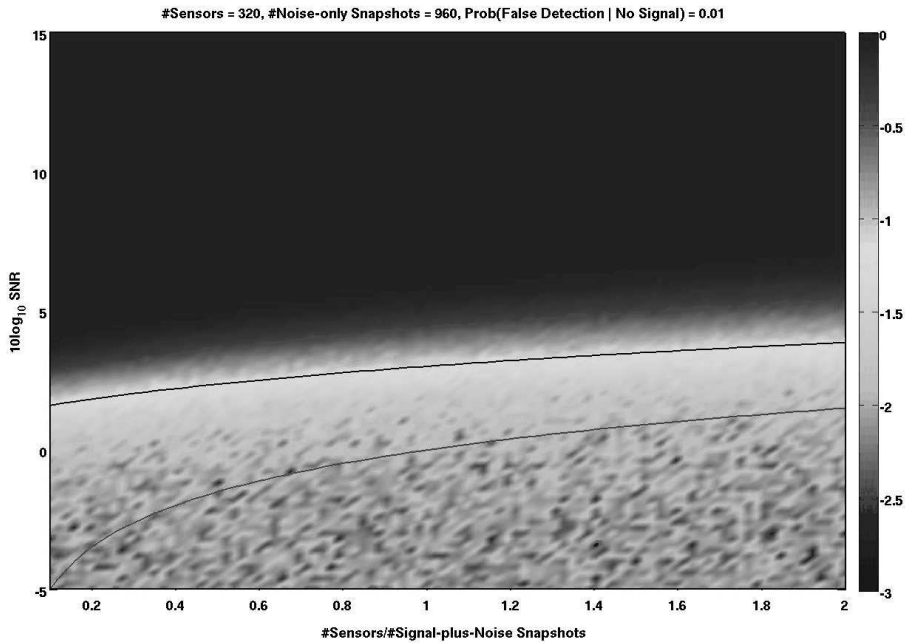}
\caption{A heat map of the log probability of signal detection using Algorithm 1 in Section \ref{sec:new algorithm}, with the significance level $\alpha$ set at $0.01$, in (eigen) SNR versus number of sensors to number of signal-plus-noise snapshots phase space. In this example, for the setting described in Theorem \ref{prop:spiked convergence} we set $n = 320$, $N = 960$ and w.l.o.g. $\bm{\Sigma} = {\bf I}$, ${\bf R} = \textrm{diag}(\lambda_{1} = 1+ \textrm{SNR},1,\ldots, 1)$ and evaluated Prob($\hat{k} = 1$) over $1000$ Monte-Carlo trials and a grid of $100$ equally spaced points in the -5 dB to 15 dB (eigen) SNR range and $100$ equally spaced points in the $c_{1} = n/m$ space by setting $m = n/c_{1}$. The values of the colormap at each of the $100 \times 1000$ faces were interpolated across each line segment and face to obtain the above plot. In the dark zone (upper half of the plot) a signal can be reliably detected  whereas in the lighter zone (lower half of the plot) the signal is statistically indistinguishable from noise as evidenced from the probability of detection being close to the significance level. The superimposed solid black line demarcates the theoretically predicted threshold  while the superimposed solid red line is the theoretically predicted threshold in the setting where the noise covariance matrix is perfectly known. The gap between the two lines thus represents the SNR loss due to noise covariance matrix estimation.}
\label{fig:threshold experiment}
\end{figure}

\section{Conclusion}\label{sec:conclusion}
Figure \ref{fig:threshold experiment} captures the fundamental statistical limit encountered when attempting to discriminate signal from noise using finite samples. Simply put, a signal whose eigen-SNR is below the detectability threshold cannot be reliably detected while a signal above the threshold can be. In settings such as wireless communications and biomedical signal processing where the signal power is controllable, our results provide a prescription for how strong it needs to be so that it can be detected. If the signal level is barely above the threshold, simply adding more sensors might actually degrade the performance because of the increased dimensionality of the system. If, however, either due to clever signal design or physics based modeling, we are able to reduce (or identify) the dimensionality of the subspace spanned by signal, then according to Figure \ref{fig:threshold experiment} the detectability threshold will also be lowered. With VLSI advances making sensors easier and cheaper to deploy, our results demonstrate exactly why the resulting gains in systemic performance will more than offset the effort we will have to invest in developing increasingly more sophisticated dimensionality reduction techniques. Understanding the fundamental statistical limits of techniques for signal detection in the setting where the noise-only sample covariance matrix is singular remains an important open problem.

\section*{Acknowledgements}
Raj Rao was supported by an Office of Naval Research Post-Doctoral Fellowship Award under grant N00014-07-1-0269. Jack Silverstein was supported by the U.S. Army Research Office under Grant W911NF-05-1-0244. R. R. thanks Arthur Baggeroer for encouragement and invaluable feedback. This material was based upon work supported by the National Science Foundation under Agreement No. DMS-0112069. Any opinions, findings, and conclusions or recommendations expressed in this material are those of the author(s) and do not necessarily reflect the views of the National Science Foundation.

We thank Alan Edelman for his encouragement and support, Interactive Supercomputing, Inc. for providing access to the Star-P parallel computing software and Sudarshan Raghunathan of Interactive Supercomputing, Inc. for his patience and support in answering our multitude of Star-P programming queries. We remain grateful to Al Davis and Chris Hill of MIT for granting us access to the Darwin Project computing cluster. Thanks to their involvement we were able to program, debug and complete the computation needed to produce Figure \ref{fig:threshold experiment} in 4 days! Without their gracious help, the computation would have taken 3 months on the latest single processor laptop. We thank Folkmar Bornemann for providing the \matlab code for computing the percentiles in Table \ref{tab:tw quantiles}.

\section{Appendix}\label{sec:appendix}

\subsection{Mathematical preliminaries}
Let for $i,j=1,2,\ldots$, $X_{ij}$ be a collection of complex valued i.i.d. random variables with $\exp X_{1\,1}=0$ and $\exp|X_{1\,1}|^2=1$. For positive integers $n$ and $m$ let ${\bf X}_n=(X_{ij})$, $i=1,2,\ldots,n$, $j=1,2,\ldots,m$.  Assume for each $n$  ${\bf T}_n$ is an $n\times n$ Hermitian nonnegative definite matrix.  The matrix $${\bf B}_n\equiv(1/m)\Th_n {\bf X}_n{\bf X}_n^*\Th_n,$$ where $\Th_n$ is any Hermitian square root of ${\bf T}_n$, can be viewed as a sample covariance matrix, formed from $m$ samples of the random vector $\Th_n {\bf X}_{\cdot 1}$ with ${\bf X}_{\cdot 1}$ denoting the first column of ${\bf X}_n$, which has ${\bf T}_n$ for its population covariance matrix.
When $n$ and $m$ are both large and on the same order of magnitude, ${\bf B}_n$ will not be near ${\bf T}_n$, due to an insufficient number of samples required for such a large dimensional random vector. However, there exist results on the eigenvalues of ${\bf B}_n$. They are limit theorems as $n\to\infty$ with $m=m(n)$ and $c_n\equiv n/m\to c$, which provide information on the eigenvalues of ${\bf T}_n$.  One result \cite{silverstein95b} is on the {\sl empirical distribution function} (e.d.f.),  $F^{B_n}$, of the eigenvalues of ${\bf B}_n$, which throughout the paper, is defined for any Hermitian $n\times n$ matrix ${\bf A}$ as
$$F^A(x)\equiv (1/n)(\text{number of eigenvalues of A $\leq x$)}.$$

The limit theorem is expressed in terms of the {\sl Stieltjes transform} of the limiting e.d.f. of the $F^{B_n}$'s, where for any distribution
function (d.f.) $G$ its Stieltjes transform, $m_G$, is defined to be
$$m_G(z)=\int\frac1{\lambda-z}dG(\lambda),\quad z\in\comp^+\equiv \{z\in\comp:\Im z>0\}.$$

There exists a one-to-one correspondence between the distribution functions (d.f.'s) and their Stieltjes transforms, due to the inversion formula
$$G(b)-G(a)=\lim_{v\to0}\int_a^b\Im m_G(x+iv)dx,$$
for $a,b$ continuity points of $G$.

The limit theorem allows the ${\bf T}_n$ to be random, only assuming as $n\to\infty$, the convergence of $F^{T_n}$ to a nonrandom proper probability distribution function $H_n$, \ie,  $H_n\equiv F^{T_n}\asarrow H$. The theorem states that with probability one, as $n\to\infty$, $F^{B_n}\darrow F$, where $F$ is nonrandom, with Stieltjes transform $m=m_F(z)$, $z\in\comp^+$ satisfying the equation
\begin{equation}\label{eq:limit theorem 1}
m=\int\frac1{t(1-c-czm)-z}dH(t),
\end{equation}
which is unique in the set $\{m\in\comp:-\frac{1-c}z+cm\in\comp^+\}$.

It is more convenient to work with the eigenvalues of the $m\times m$ matrix $(1/m){\bf X}_n^{'}{\bf T}_n{\bf X}_n$, whose eigenvalues differ from those of ${\bf B}_n$ by $|n-m|$ zero eigenvalues.  Indeed, with $I_A$ denoting the indicator function on the set $A$ we have the exact relationship
$$F^{(1/m)X_n^*T_nX_n}(x)=(1-c_n))I_{[0,\infty)}(x)+c_nF^{B_n}(x)$$
$$\darrow =(1-c))I_{[0,\infty)}(x)+cF(x)\equiv F^{c,H}(x)$$
almost surely, implying
\begin{equation}\label{eq:identity above 2}
m_{F^{c,H}}(z)=-(1-c)/z+cm_F(z).
\end{equation}

Upon substituting $m_{F^{c,H}}$ into (\ref{eq:limit theorem 1}) we find that for $z\in\comp^+$
$m=m_{F^{c,H}}(z)$ solves the equation
\begin{equation}\label{eq:jack 2}
z=-\frac1{m}+c\int\frac{\lambda}{1+\lambda m}dH(\lambda),
\end{equation}
and is unique in $\comp^+$.  Thus we have an explicit inverse for
$m_{F^{c,H}}$.

Qualitative properties of $F^{c,H}$ have been obtained in \cite{silverstein95a}, most notably the fact that on $(0,\infty)$ $F^{c,H}$ has a continuous derivative.  The paper \cite{silverstein95a} also shows how intervals outside
the support of $F^{c,H}$ can be determined from the graph of (\ref{eq:jack 2}) for $m\in\real$.

Let $S_G$ denote the support of the d.f. $G$, $S'_G$ its complement, and define $x_{c,H}=x_{c,H}(m)$ to be (\ref{eq:jack 2}) with $m\in\real$.  Intuitively, on $S'_{F^{c,H}}$ $m_{F^{c,H}}$ is well defined and increasing.  Therefore it is invertible on each interval in $S'_{F^{c,H}}$, its inverse, namely $x_{c,H}$, is also increasing.   The details are stated in the following.

\begin{lemma}[Theorems 4.1, 4.2 of \cite{silverstein95a}]\label{lem:jack 1}
 If $x\in S'_{F^{c,H}}$, then $m=m_{F^{c,H}}$ satisfies (1)
$m\in\real\backslash\{0\}$, (2) $-m^{-1}\in S'_H$, and (3)
$\frac{d}{dm}x_{c,H}(m)>0$.
Conversely, if $m$ satisfies (1)--(3), then $x=x_{c,H}(m)\in S'_{F^{c,H}}$.
\end{lemma}

In simple terms $S'_{F^{c,H}}$ is comprised of the range of values where $x_{c,H}$ is increasing.

Another result which will be needed later is the following.

\begin{lemma}[Theorem 4.3 of \cite{silverstein95a}]
\label{lem:jack 2}
Suppose each $m$ contained in the interval $[m_1,m_2]$ satisfies (1) and (2) of Lemma \ref{lem:jack 1}, and $\frac{d}{dm}x_{c,H}(m_i)\ge0$ for $i=1,2$.  Then $\frac{d}{dm}x_{c,H}(m)>0$
for all $m\in(m_1,m_2)$.
\end{lemma}

Limiting eigenvalue mass at zero is also derived in \cite{silverstein95a}.
It is shown that
\begin{equation}\label{eq:jack 3}
F(0)=\begin{cases}
\hfill H(0),\hfill &c(1-H(0))\leq1,\\
\hfill 1-c^{-1},\hfill&c(1-H(0))>1.
\end{cases}
\end{equation}

\subsection{Support of eigenvalues}

Since the convergence in distribution of $F^{B_n}$ only addresses how proportions of eigenvalues behave, understanding the possible appearance or non-appearance of eigenvalues in $S'_{F^{c,H}}$ requires further work.

The question of the behavior of the largest and smallest eigenvalues when ${\bf T}_n={\bf I}$ has been answered by Yin, Bai, and Krishnaiah in \cite{yinbai88a}, and Bai and Yin in \cite{bai93c}, respectively, under the additional assumption $\exp|{\bf X}_{1\,1}|^4<\infty$: the largest eigenvalue and $\min(n,m)^{\text{th}}$ largest eigenvalue of
$(1/m){\bf X}_n{\bf X}_n^*$ converge a.s. to $(1+\sqrt c)^2$ and $(1-\sqrt c)^2$
respectively, matching the support, $[(1-\sqrt c)^2,(1+\sqrt c)^2]$ of $F$ on $(0,\infty)$.  More on $F$ when ${\bf T}_n={\bf I}$ will be given later.

For general ${\bf T}_n$, restricted to being bounded in spectral norm, the non-appearance of eigenvalues in $S'_{F^{c,H}}$ has been proven by Bai and Silverstein in \cite{bai98z}.
Moreover, the separation of eigenvalues across intervals in $S'_{F^{c,H}}$, mirrors exactly the separation of eigenvalues over corresponding
intervals in $S'_H$ \cite{bai99z}.  The results are summarized below.

\begin{theorem}\label{th:jack 1}
Assume additionally $\exp|{\bf X}_{1\,1}|^4<\infty$ and the ${\bf T}_n$ are nonrandom and are bounded in spectral norm for all $n$.

Let $F^{c_n,H_n}$ denote the ``limiting" e.d.f. associated with $(1/m){\bf X}_n^*{\bf T}_n{\bf X}_n$, in other words, $F^{c_n,H_n}$ is the d.f.
having Stieltjes transform with inverse (\ref{eq:jack 2}), where $c,H$ are replace by $c_n,H_n$.

Assume the following condition:
\begin{itemize}
\item{(*)} Interval $[a,b]$ with $a>0$ lies in an open interval outside the support of $F^{c_n,H_n}$ for all large $n$.
\end{itemize}
Then $\P(\text{no eigenvalue of ${\bf B}_n$ appears in $[a,b]$ for all large $n$})=1$.

For $n\times n$ Hermitian non-negative definite matrix ${\bf A}$, let $\lambda_k^A$ denote the $k^{\text{th}}$ largest eigenvalue of $A$.  For notational convenience, define  $\lambda_0^{A}=\infty$ and
 $\lambda_{n+1}^{A}=0$.

(i) If $c(1-H(0))>1$, then $x_0$, the smallest value in the support of $F^{c,H}$, is positive, and with probability 1, $\lambda_m^{B_n}\to x^0$ as $n\to\infty$.

(ii) If $c(1-H(0))\leq1$, or  $c(1-H(0))>1$ but $[a,b]$ is not contained in $[0,x_0]$, then $m_{F^{c,H}}(b)<0$, and for all $n$ large there is an index $i_n$ for which
\begin{equation}\label{eq:jack 4}
\lambda_{i_n}^{T_n}>-1/m_{F^{c,H}}(b)\quad\text{and}\quad \lambda_{i_n+1}^{T_n}<-1/m_{F^{c,H}}(a).
\end{equation}
Then $\P(\lambda_{i_n}^{B_n}>b\text{ and $\lambda_{i_n+1}^{B_n}<a$ for all large $n$})=1$.
\end{theorem}
\begin{proof}
See proof of  Theorems 1.1 in \cite{bai98z,bai99z}).
\end{proof}

The behavior of the extreme eigenvalues of $(1/m){\bf X}_n{\bf X}_n^*$ leads to the following corollary of Theorem \ref{th:jack 1}.

\begin{corollary}
If $\lambda_1^{T_n}$ converges to the largest number in the support of $H$, then $\lambda_1^{B_n}$ converges a.s to the largest number in the support of $F$.  If $\lambda_n^{T_n}$ converges to the smallest number in the support of $H$, then $c\leq1$ ($c>1$) implies $\lambda_n^{B_n}$ ($\lambda_n^{(1/m)X_n^*T_nX_n}$) converges a.s. to the smallest number in the support of $F$ ($F^{c,H}$).
\end{corollary}

In Theorem \ref{th:jack 1}, Case (i) applies when $n>m$, whereby the rank of ${\bf B}_n$ would be at most $m$, the conclusion asserting, that with probability 1, for all $n$ large, the  rank is equal to $m$.  From Lemma \ref{lem:jack 1}, Case (ii) of Theorem \ref{th:jack 1} covers all intervals in $S'_{F^{c,H}}$ on $(0,\infty)$ resulting from intervals on $(-\infty,0)$ where $x_{c,H}$ is increasing.  For all $n$ large $x_{c_n,H_n}$ is increasing on
$[m_{F^{c_n,H_n}}(a),m_{F^{c_n,H_n}}(b)]$, which, from inspecting the vertical asymptotes of $x_{c_n,H_n}$ and Lemma \ref{lem:jack 1}, must be due to the existence of $\lambda_{i_n}^{T_n}$, $\lambda_{i_n+1}^{T_n}$
satisfying (\ref{eq:jack 4}).

Theorem \ref{th:jack 1}  easily extends to random ${\bf T}_n$, independent of $\{{\bf X}_{ij}:i,j\ge1\}$  with the aid of Tonelli's Theorem \cite[pp. 234]{billingsley95a}, provided the condition (*) on $[a,b]$ is strengthened to:

\begin{itemize}
\item{(**)} With probability 1 for all $n$ large $[a,b]$ (nonrandom) lies in an open interval outside the support of $F^{c_n,H_n}$.
\end{itemize}
Indeed, let $T$ denote the probability space generating $\{T_n\}$, $X$ the probability space generating  $\{X_{ij}:i,j\ge1\}$. Let their
respective measures be denoted by $\P_T$,$\P_X$, the product measure on
$T\times X$ by $\P_{T\times X}$.  Consider, for example
in case (ii), we define
$$A=\{\lambda_{i_n}^{B_n}>b\text{ and $\lambda_{i_n+1}^{B_n}<a$
for all large $n$}\}.$$
Let $t\in T$ be an element of the event defined in (**).  Then by
Theorem \ref{th:jack 1} $I_A((t,x))=1$ for all $x$ contained in a subset of $X$ having
probability 1.
Therefore, by Tonelli's theorem
$$\P(A)=\int I_A(t,x)dP_{T\times X}((t,x))=
\int\left[\int I_A(t,x)dP_X(x)\right]dP_T(t)=\int 1dP_T(t)=1.$$

Consider now case (ii) of Theorem \ref{th:jack 1} in terms of the corresponding interval outside the support of $H$ and the $H_n$'s. By Lemma \ref{lem:jack 1}  and condition (*), we have the
existence of an $\epsilon>0$ such that
$0\notin[m_{F^{c,H}}(a)-\epsilon,m_{F^{c,H}}(b)+\epsilon]$, and
for all $n$ large
\begin{equation}\label{eq:jack 5}
\frac{d}{dm}x_{c_n,H_n}(m)=\frac1{m^2}\left(1-c_n\int\frac{(\lambda m)^2}
{(1+\lambda m)^2}dH_n(\lambda)\right)>0,\quad m\in
[m_{F^{c,H}}(a)-\epsilon,m_{F^{c,H}}(b)+\epsilon].
\end{equation}

 Let $t_a=-1/m_{F^{c,H}}(a)$, $t_b=-1/m_{F^{c,H}}(b)$.  Then by Lemma \ref{lem:jack 1} we have the existence of an $\epsilon'>0$ for which $t_a-\epsilon'>0$ and $[t_a-\epsilon',t_b+\epsilon']\subset S'_{H_n}$ for all $n$ large.
Moreover, by (\ref{eq:jack 5}) we have for all $n$ large
\begin{equation}\label{eq:jack 6}
c_n\int\frac{\lambda^2}{(\lambda-t)^2}dH_n(\lambda)<1,\quad t\in
[t_a-\epsilon',t_b+\epsilon'].
\end{equation}
Necessarily, $\lambda_{i_n}^{T_n}>t_b+\epsilon'$ and $\lambda_{i_n+1}^{T_n}<t_a-\epsilon'$.

Notice the steps can be completely reversed, that is, beginning with an interval $[t_a,t_b]$, with $t_a>0$, lying in an open interval in $S'_{H_n}$ for all $n$ large and satisfying (\ref{eq:jack 6}) for some $\epsilon'>0$, will yield $[a,b]$, with $a=x_{c,H}(-1/t_a)$, $b=x_{c,H}(-1/t_b)$, satisfying condition (*).  Case (ii) applies, since $[a,b]$ is within the range of $x_{c,H}(m)$ for $m<0$.  If $c(1-H(0))>1$, then we would have $a>x_0$.

\subsection{Behavior of spiked eigenvalues}
Suppose now the ${\bf T}_n$'s are altered, where a finite number of eigenvalues are  interspersed between the previously adjacent eigenvalues $\lambda^{T_n}_{i_n+1}$ and $\lambda^{T_n}_{i_n}$.  It is clear
that the limiting $F$ will remain unchanged.  However, the graph of $x_{c_n,H_n}$ on $(-1/\lambda^{T_n}_{i_n+1},-1/\lambda^{T_n}_{i_n})$ will now contain vertical asymptotes.  If the graph remains increasing on two intervals for all $n$ large, each one between successive asymptotes, then because of Theorem \ref{th:jack 1}, with probability one, eigenvalues of the new
${\bf B}_n$ will appear in $S'_{F^{c,H}}$ for all $n$ large.

Theorem \ref{th:jack 2} below shows this will happen when a ``sprinkled", or ``spiked" eigenvalue lies in $(t_a,t_b)$.  Theorem \ref{th:jack 3} provides a converse, in the sense that any isolated eigenvalue of $B_n$ must be due to a spiked eigenvalue, the
absence of which corresponds to case (ii) of Theorem \ref{th:jack 1}.

Theorem \ref{th:jack 2}, below, allows the number of spiked eigenvalues to grow with $n$, provided it remains $o(n)$.

\begin{theorem}\label{th:jack 2}
Assume in additon to the assumptions in Theorem \ref{th:jack 1}
on the ${\bf X}_{ij}$ and ${\bf T}_n$:

\indent(a)  There are $\ell=o(n)$ positive eigenvalues of ${\bf T}_n$ all converging uniformly to $t'$, a positive number.
Denote by $\hat H_n$ the e..d.f. of the $n-\ell$ other eigenvalues of ${\bf T}_n$.
\hfil\break
\indent(b)  There exists positive $t_a<t_b$ contained in an interval $(\alpha,\beta)$ with $\alpha>0$ which is outside the support of $\hat H_n$ for all large $n$, such that for these $n$
$$c_n\int\frac{\lambda^2}{(\lambda-t)^2}d\hat H_n(\lambda)\leq 1 $$
{\sl for $t=t_a,t_b$.}\hfil\break
\indent(c) {\sl $t'\in(t_a,t_b)$.}\hfil\break

Suppose $\lambda^{T_n}_{i_n},\ldots, \lambda^{T_n}_{i_n+\ell-1}$ are the eigenvalues stated in (a).  Then,
with probability one
\begin{equation}\label{eq:jack 7a}
\lim_{n\to\infty}\lambda^{B_n}_{i_n}=\cdots=\lim_{n\to\infty}
 \lambda^{B_n}_{i_n+\ell-1}
=t'\left(1+c\int\frac{\lambda}{t'-\lambda}
dH(\lambda)\right).
\end{equation}
\end{theorem}
\begin{proof}
For $m\in[-1/t_a,-1/t_b]\cap\{-1/t'\}^c$, we have
$$x_{c_n,H_n}(m)=
-\frac1m+c_n\left(\frac1n\sum_{j=i_n}^{i_n+\ell-1}\frac{\lambda_j^{T_n}}
{1+\lambda_j^{T_n}m}+\frac{n-\ell}n
\int\frac{\lambda}{1+\lambda m}d\hat H_n(\lambda)\right).$$

By considering continuity points of $H$ in $(\alpha,\beta)$ we see that $H$ is constant on this interval, and consequently, this interval is also contained in $S'_H$.

Because of (b) we have $\frac{d}{dm}x_{c,H}(m)\ge0$ for $m=-1/t_a,-1/t_b$ (recall (\ref{eq:jack 5}),(\ref{eq:jack 6})).

By Lemma \ref{lem:jack 2} we therefore have $\frac{d}{dm}x_{c,H}(m)>0$ for all $m\in(-1/t_a,-1/t_b)$. Thus we can find $[\underline t_a,\underline t_b]\subset[t_a,t_b]$ and $\delta
>0$, such that $t'\in(\underline t_a,\underline t_b)$ and for all $n$ large $\frac{d}{dm}x_{c_n,\hat H_n}(m)\ge\delta$ for all $m\in[-1/\underline t_a,-1/\underline t_b]$.

It follows that for any positive $\epsilon$ sufficiently small, there exist positive $\delta'$ with $\delta'\leq\epsilon$, such that, for all $n$ large,
both $[-1/t'-\epsilon-\delta',
-1/t'-\epsilon]$, and $[-1/t'+\epsilon,-1/t'+\epsilon+\delta']$:
\begin{itemize}
\item{1)} are
contained in $[-1/\underline t_a,-/\underline t_b]$, and
\item{2)}
$\frac{d}{dm}x_{c_n,H_n}(m)>0$
for all $m$ contained in these two intervals.
\end{itemize}
Therefore, by Lemma \ref{lem:jack 1}, for all $n$ large,
$[x_{c_n,H_n}(-1/t'-\epsilon-\delta'),
x_{c_n,H_n}(-1/t'-\epsilon)]$ and $[x_{c_n,H_n}(-1/t'+\epsilon),
x_{c_n,H_n}(-1/t'+\epsilon+\delta')]$ lie
outside the support of $F^{c_n,H_n}$.
Let $a_L=x_{c,H}(-1/t'-\epsilon-\tfrac23\delta')$,
$b_L=x_{c,H}(-1/t'-\epsilon-\tfrac13\delta')$,
$a_R=x_{c,H}(-1/t'+\epsilon+\tfrac13\delta')$,
and $b_R=x_{c,H}(-1/t'+\epsilon+\tfrac23\delta')$.
Then for all $n$ large
\begin{multline}
[a_L,b_L]\subset(x_{c,H}(-1/t'-\epsilon-\tfrac56\delta'),
x_{c,H}(-1/t'-\epsilon-\tfrac16\delta'))\hfill\\ \hfill\subset
[x_{c_n,H_n}(-1/t'-\epsilon-\delta'),x_{c_n,H_n}(-1/t'-\epsilon)]
\end{multline}
and
\begin{multline}
[a_R,b_R]\subset(x_{c,H}(-1/t'+\epsilon+\tfrac16\delta'),
x_{c,H}(-1/t'+\epsilon+\tfrac56\delta'))\hfill\\ \hfill\subset
[x_{c_n,H_n}(-1/t'+\epsilon),
x_{c_n,H_n}(-1/t'+\epsilon+\delta')].
\end{multline}

It follows then that $[a_L,b_L]$, $[a_R,b_R]$ each lie in an open interval in $S'_{F^{c_n,H_n}}$ for all $n$ large.  Moreover $m_{F^{c,H}}(b_R)<0$. Therefore, case (ii) of Theorem \ref{th:jack 1} applies and we have
$$\P(\lambda^{B_n}_{i_n}<a_R\quad\text{and}\quad\lambda^{B_n}_{i_n+\ell-1}>b_L\quad\text{for all
$n$ large})=1.$$
Therefore, considering a countable collection of $\epsilon$'s
converging to zero,
we conclude that, with probability 1
$$
\lim_{n\to\infty}\lambda^{B_n}_{i_n}=\lim_{n\to\infty}\lambda^{B_n}_{i_n+\ell-1}
=x_{c,H}(-1/t')=(\ref{eq:jack 7a}).
$$
\end{proof}

\begin{theorem}\label{th:jack 3}
Assume, besides the assumptions in Theorem \ref{th:jack 1}, there is an eigenvalue of ${\bf B}_n$ which converges in probability to a nonrandom positive number, $\lambda'\in S'_F$.  Let interval $[a,b]\in S'_F$, with $a>0$,
be such that $\lambda'\in(a,b)$, and let $t_a=-1/m_{c,H}(a)$, $t'=-1/m_{c,H}(\lambda')$, $t_b=-1/m_{c,H}(b)$ (finite by Lemma \ref{lem:jack 1}).
Then $0<t_a<t'<t_b$, implying (c) of Theorem \ref{th:jack 2}. Let $\ell=\ell(n)$
denote the number of eigenvalues of ${\bf T}_n$ contained in $[t_a,t_b]$ and let $\hat H_n$ denote the e.d.f. of the other $n-\ell$ eigenvalues of ${\bf T}_n$.  Then $\ell=o(n)$ and (b) of Theorem \ref{th:jack 2} is true.  If $\ell$ remains bounded, then (a) of Theorem \ref{th:jack 2} also holds.
\end{theorem}
\begin{proof}
By Lemma \ref{lem:jack 1}, $[t_a,t_b]\in S'_H$, and for a suitable positive $\epsilon$, $x_{c,H}$ is increasing on $[m_{c,H}(a)-\epsilon,m_{c,H}(b)+\epsilon]$, which does not contain 0.

Therefore $t_a<t'<t_b$.  If $c(1-H(0))>1$, that is, case (i) of Theorem \ref{th:jack 1} holds, then $a>x_0$, since $x_0$ is the almost sure limit of $\lambda_m^{B_n}$ so $\lambda'$ cannot be smaller than it, and necessarily $x_0\in S_F$. Therefore $m_{c,H}(b)<0$, so that $0<t_a$.

It must be the case that only $o(n)$ eigenvalues of $t_n$ lie in $[t_a,t_b]$, since otherwise $[t_a,t_b]$ would not be outside the support of $H$.  We have then $\hat H_n\darrow H$ as $n\to\infty$, so
from the dominated convergence theorem we have
$\frac{d}{dm}x_{c_n,\hat H_n}(m)\to\frac{d}{dm}x_{c,H}(m)$ for all
$m\in[m_{c,H}(a)-\epsilon,m_{c,H}(b)+\epsilon]$, implying for
all $n$ large $\frac{d}{dm}x_{c_n,\hat H_n}(m)>0$ for all
$m\in[m_{c,H}(a),m_{c,H}(b)]$. Therefore (b) is true

We assume now that $\ell$ is bounded. Suppose (a) does not hold.  Then we could find a subsequence $\{n_j\}$ of the natural numbers for which
$\ell'=\ell'(n)$ of the $\ell$ eigenvalues converge to a $\underline t'\neq t'$, the remaining $\ell-\ell'$, if positive, eigenvalues remaining a positive
distance $d$ from $\underline t'$.  Replace $\{{\bf T}_n\}$ with $\{{\bf T}_n'\}$ which matches the
original sequence when $n=n_j$ and  for $n\neq n_j$, ${\bf T}_n'$ has $\ell'$  eigenvalues equal to $\underline t'$, with the remaining $\ell-\ell'$, again, if positive, eigenvalues of ${\bf T}_n'$ at least $d$ away from $\underline t'$. Then we have by Theorem \ref{th:jack 1}, (\ref{eq:jack 7a}), with $t'$ replaced by $\underline t'$,  holding for
$\ell'$ of the eigenvalues of $(1/m){{\bf T}_n'}^{1/2}{\bf X}_n{\bf X}_n^*
{T_n'}^{1/2}$.  Thus, on $\{n_j\}$, we have the almost sure convergence of $\ell$ eigenvalues of ${\bf B}_n$ to $x_{c,H}(-1/\underline t')\in[a,b]$ which, because $x_{c,H}(-1/t)$
is an increasing function, does not equal $\lambda'=x_{c,H}(-1/t')$. This contradicts the assumption of convergence in probability to eigenvalues to only one number, namely $\lambda'$.  Therefore (a) holds.
\end{proof}

\subsection{Behavior of extreme eigenvalues}

Consider now $t'$ lying on either side of the support of $\hat H$. Let $\hat\lambda_n^{\min}$ and $\hat\lambda_n^{\max}$ denote, respectively,
the smallest and largest numbers in the support of $\hat H_n$ Notice that $g_n(t)\equiv c_n\int\frac{\lambda^2}{(\lambda-t)^2}d\hat H_n(t)$
is decreasing for $t>\hat\lambda_n^{\max}$, and if $\hat\lambda_n^{\min}>0$, $g_n$ is increasing on $(0,\hat\lambda_n^{\min})$.

Therefore, if for all $n$ large, $t'>\hat\lambda_n^{\max}$, it is necessary and sufficient to find a $t_a\in(\hat\lambda_n^{\max},t')$ for which $g(t_a)\leq1$ in order for (\ref{eq:jack 7a})
to hold.  Similarly, if for all $n$ large $t'\in(0,\hat\lambda_n^{\min})$, then it is necessary and sufficient to find a $t_b\in(0,t')$ for which $g_n(t_b)\leq1$
in order for (\ref{eq:jack 7a}) to hold.  Notice if $c(1-H(0))>1$ then $g_n(t)>1$ for all $t\leq\hat\lambda_n^{\min}$ and all $n$ large.

Let for d.f. $G$ with bounded support, $\lambda_G^{\max}$ denote the largest number in $S_G$. If there is a $\tau>\lambda_H^{\max}$  for which $g(\tau)= c\int\frac{\lambda^2}{(\lambda-t)^2}dH(t)=1$, and if
$\limsup_n\hat\lambda_n^{\max}<\tau$, then $\tau$ can be used as a threshold for $t'\in(\limsup_n\hat\lambda_n^{\max},\infty)$.  Indeed,
by the dominated convergence theorem, $\lim_{n\to\infty}g_n(t')=g(t')$. Therefore, if $t'>\tau$, conditions (b) and (c) of Theorem \ref{th:jack 2} hold, with $t_a=\tau$, and $t_b$ any arbitrarily large number.

On the other hand, suppose $\lambda_{i_n}^{T_n}, \ldots,\lambda_{i_n+\ell-1}^{T_n}$, where $\ell$ remains bounded,
are the eigenvalues of ${\bf T}_n$ approaching the interval
$(\limsup_n\hat\lambda_n^{\max},\tau]$. Then by Theorem \ref{th:jack 3}, for any
$\epsilon>0$ with
probability one,
none of $\lambda_{i_n}^{B_n},\ldots,\lambda_{i_n+\ell-1}^{B_n}$
can remain in $(\lambda_F^{\max}+\epsilon,\infty)$ with for all $n$ large.

Also, since the largest $i_n+\ell-1$ eigenvalues of ${\bf T}_n$ must be $o(n)$ (otherwise, $H$ would have additional mass on $[\lambda_H^{\max},\infty)$), $\lambda_{i_n}^{B_n},\ldots,
\lambda_{i_n+\ell-1}^{T_n}$ must all converge a.s. to $\lambda_F^{\max}$.

Similar results can be obtained for the interval to the left of $S_F$

As in Theorem \ref{th:jack 1} Tonelli's Theorem can easily be applied to establish equivalent
results when ${\bf T}_n$'s are random and independent of ${\bf X}$.

\subsection{The eigenvalues of the multivariate F matrix}\label{sec:appendix F}
Let ${\bf Y}_{ij}$ be another collection of i.i.d.
random variables (not necessarily having the same distribution as
the ${\bf X}_{ij}$'s), with $\exp {\bf Y}_{1\,1}=0$, $\exp|{\bf Y}_{1\,1}|=1$,
 $\exp|{\bf Y}_{1\,1}|^4<\infty$, and independent of the ${\bf X}_{ij}$'s.

We form the $n\times N$ matrix ${\bf Y}_n=(Y_{ij})$, $i=1,2,\ldots,n$, $j=1,2,\ldots,N$ with $N=N(n)$,
$n<N$, and $c^1_n\equiv n/N\to c_1\in(0,1)$ as $n\to\infty$.

Let now ${\bf T}_n=((1/N){\bf Y}_n{\bf Y}_n^*)^{-1}$, whenever the inverse exists.

From Bai and Yin's work \cite{bai93c} we know that with probability 1, for all $n$ large, ${\bf T}_n$ exists with
$\lambda_1^{T_n}\to (1-\sqrt {c_1})^{-2}$.  Whenever $\lambda_n^{(1/N){\bf Y}_n{\bf Y}_n^*}=0$ define ${\bf T}_n$ to be ${\bf I}$.

The matrix ${\bf T}_n(1/N){\bf X}_n{\bf X}_n^*$, typically called a multivariate $F$ matrix, has the same eigenvalues as ${\bf B}_n$.  Its limiting e.d.f. has density on $(0,\infty)$  given by
$$f_{c,c_1}(x)=\frac{(1-c_1)\sqrt{(x-b_1)(b_2-x)}}{2\pi x(xc_1+c)}
\quad b_1<x<b_2,$$
where
$$b_1=\left(\frac{1-\sqrt{1-(1-c)(1-c_1)}}{1-c_1}\right)^2,\quad
b_2=\left(\frac{1+\sqrt{1-(1-c)(1-c_1)}}{1-c_1}\right)^2.$$
When $c\in(0,1]$, there is no mass at $0$, whereas for $c>1$ $F$ has
mass $(1-(1/c))$ at $0$ \cite{silverstein85a}.

We are interested in the effect on spikes on the right side of the support of the $H_n$.

Because of the corollary to Theorem \ref{th:jack 1}, we know $\lambda_1^{B_n}\to b_2$ a.s. as $n\to\infty$.  We proceed in computing
the function
$$g(t)= c\int\frac{\lambda^2}{(\lambda-t)^2}dH(t).$$

We will see that it is unnecessary to compute the limiting e.d.f. of ${\bf T}_n$.
It suffices to know the limiting Stieltjes transform of $F^{(1/N)Y_nY_n^*}$.

Let $H_1$ denote the limiting e.d.f. of $F^{(1/N)Y_nY_n^*}$.  We have
$$g(t)=c\int\frac{(1/\lambda)^2}{(t-1/\lambda)^2}dH_1(\lambda)
=c\int\frac1{(\lambda t-1)^2}dH_1(\lambda)=ct^{-2}\int\frac1{(\lambda-(1/t))^2}
dH_1(\lambda)$$
$$=t^{-2}\frac{d}{dx}m_{H_1}(x)\biggr|_{x=(1/t)}.$$

We use (\ref{eq:jack 2}) to find $m_{F^{c_1,I_{[1,\infty)}}}$:
$$z=-\frac1m+c_1\frac1{1+m}\quad\Leftrightarrow\quad zm^2+(z+1-c_1)m+1=0.$$
$$\Leftrightarrow\quad m=\frac{-z-1+c_1\pm\sqrt{(z+1-c_1)^2-4z}}{2z}$$
(the sign depending on with branch of the square root is taken).
$$=\frac{-z-1+c_1\pm\sqrt{(z-(1-\sqrt{ c_1})^2)(z-(1+\sqrt{c_1})^2)}}{2z}.$$

From the identity in (\ref{eq:identity above 2}) we find that
$$m_{H_1}(z)=\frac{-z+1-c_1\pm\sqrt{(z-(1-\sqrt{c_1})^2)(z-(1+\sqrt{c_1})^2)}}
{2c_1z}.$$

As mentioned earlier the support of $H_1$ is $[(1-\sqrt{c_1})^2,
(1+\sqrt{c_1})^2]$.
We need $g(t)$ for $t>(1-\sqrt{c_1})^{-2}$, so we
need $m_{H_1}(x)$ for $x\in(0,(1-\sqrt{c_1})^2)$.

Since $0\in S'_{H_1}$,  $m_{H_1}(0)$ exists and is real, which dictates what
sign is taken on $(0,(1-\sqrt{c_1})^2)$. We find that, on this interval
\begin{equation}\label{eq:jack 8}
m_{H_1}(x)=\frac{-x+1-c_1-\sqrt{(x-(1-\sqrt{c_1})^2)(x-(1+\sqrt{c_1})^2)}}
{2c_1x},
\end{equation}
and using the fact that the discriminant equals $x^2-2x(1+c_1)+(1-c_1)^2$,
$$\frac{d}{dx}m_{H_1}(x)=
-\frac1{2c_1x^2}\left((1-c_1)+
\frac{x(1+c_1)-(1-c_1)^2}{\sqrt{(x-(1-\sqrt{c_1})^2)(x-(1+\sqrt{c_1})^2)}}
\right).$$

We therefore find that for $t>(1-\sqrt{c_1})^{-2}$
$$g(t)=\frac{c}{2c_1}\left(-(1-c_1)+\frac{t(1-c_1)^2-(1+c_1)}
{\sqrt{(1-t(1-\sqrt{c_1})^2)(1-t(1+\sqrt{c_1})^2)}}\right).$$

We see that the equation $g(t)=1$ leads to the following quadratic equation
in $t$:
$$(1-c_1)^2\alpha t^2-2(1+c_1)\alpha t+\alpha-c^2=0,\quad\text{where }
\alpha=c_1+c-cc_1,$$
giving us
$$t=\frac{(1+c_1)\alpha+\sqrt{(1+c_1)^2\alpha^2-(1-c_1)^2\alpha(\alpha-c^2)}}
{(1-c_1)^2\alpha},$$
The positive sign in front of the square root being correct due to
$$\frac{(1+c_1)}{(1-c_1)^2}=\frac{(1+c_1)}
{(1-\sqrt{c_1})^2(1+\sqrt{c_1})^2}<\frac1{(1-\sqrt{c_1})^2}.$$
Reducing further we find the threshold, $\tau$,  to be
\begin{equation}\label{eq:jack 9}
\tau=\frac{(1+c_1)\alpha+\sqrt{\alpha}\sqrt{4\alpha-c_1+(1-c_1)^2c^2}}
{(1-c_1)^2\alpha}=\frac{(1+c_1)\alpha+\sqrt{\alpha}(2c_1+c(1-c_1))}
{(1-c_1)^2\alpha}.
\end{equation}

We now compute the right hand side of (\ref{eq:jack 7a}). We have for $t'\ge\tau$
\begin{multline}\label{eq:jack 10}
 t'\left(1+c\int\frac{\lambda}{t'-\lambda}dH(\lambda)\right)
=t'\left(1+c\int\frac{1/\lambda}{t'-1/\lambda}dH_1(\lambda)\right)
=t'(1+c{t'}^{-1}m_{H_1}(1/t'))\\
=\frac{t'(2c_1+c(1-c_1))-c-
c\sqrt{(1-t'(1-\sqrt{c_1})^2)(1-t'(1+\sqrt{c_1})^2)}}{2c_1}\equiv\lambda(t').
\end{multline}
A straightforward (but tedious) calculation will yield $\lambda(\tau)=b_2$.

Using the results from the previous section, we have proved the following:

\begin{theorem}\label{th:jack 4}
Assume in addition to the assumptions in Theorem \ref{th:jack 1} on the ${\bf X}_{ij}$

(a) the ${\bf T}_n$, possibly random, are independent of the ${\bf X}_{ij}$, with
$F^{T_n}\darrow H$, a.s. as $n\to\infty$, $H$ being the limiting e.d.f. of
$F^{((1/N)Y_nY_n^*)^{-1}}$, defined above.

(b) Almost surely, there are $\ell$ (remaining finite for each realization)
eigenvalues of ${\bf T}_n$ converging to nonrandom $t'>(1-\sqrt{c_1})^{-2}$,
as $n\to\infty$.  Denote
by $\hat H_n$ the e.d.f. of the $n-\ell$ other eigenvalues of ${\bf T}_n$.

(c) With $\hat\lambda_n^{\max}$ defined to be the largest number in the support
of $\hat H_n$, with probability one, $\limsup_n\hat\lambda_n^{\max}<\tau$
the threshold defined in (\ref{eq:jack 9}).

Suppose $\lambda_{i_n}^{T_n},\ldots,\lambda_{i_n+\ell-1}^{T_n}$ are the eigenvalues stated in (b) of Theorem \ref{th:jack 2}.
Then, with the function $\lambda(\cdot)$ defined in (\ref{eq:jack 10}), with probability one

$$\lim_{n\to\infty}\lambda_{i_n}^{B_n}=\cdots=
\lim_{n\to\infty}\lambda_{i_n+\ell-1}^{B_n}=\begin{cases}
\hfill\lambda(t'),\hfill &\text{ if } t'>\tau\\ \hfill b_2,\hfill
&\text{ if } t'\in(\limsup_n\hat\lambda_n^{\max},\tau].\end{cases}$$
\end{theorem}

Note: From Theorem \ref{th:jack 2}, when $t'>\tau$ the result can allow $\ell=o(n)$.

%% file: IEEE_colored.bbl
\def\cprime{$'$} \def\cprime{$'$}
\begin{thebibliography}{10}
\providecommand{\url}[1]{#1}
\csname url@rmstyle\endcsname
\providecommand{\newblock}{\relax}
\providecommand{\bibinfo}[2]{#2}
\providecommand\BIBentrySTDinterwordspacing{\spaceskip=0pt\relax}
\providecommand\BIBentryALTinterwordstretchfactor{4}
\providecommand\BIBentryALTinterwordspacing{\spaceskip=\fontdimen2\font plus
\BIBentryALTinterwordstretchfactor\fontdimen3\font minus
  \fontdimen4\font\relax}
\providecommand\BIBforeignlanguage[2]{{%
\expandafter\ifx\csname l@#1\endcsname\relax
\typeout{** WARNING: IEEEtran.bst: No hyphenation pattern has been}%
\typeout{** loaded for the language `#1'. Using the pattern for}%
\typeout{** the default language instead.}%
\else
\language=\csname l@#1\endcsname
\fi
#2}}

\bibitem{golub96a}
G.~H. Golub and C.~F. Van~Loan, \emph{Matrix computations}, 3rd~ed., ser. Johns
  Hopkins Studies in the Mathematical Sciences.\hskip 1em plus 0.5em minus
  0.4em\relax Baltimore, MD: Johns Hopkins University Press, 1996.

\bibitem{maris03a}
E.~Maris, ``A resampling method for estimating the signal subspace of
  spatio-temporal {EEG}/{MEG} data,'' \emph{Biomedical Engineering, IEEE
  Transactions on}, vol.~50, no.~8, pp. 935--949, Aug. 2003.

\bibitem{sekihara97a}
K.~Sekihara, D.~Poeppel, A.~Marantz, H.~Koizumi, and Y.~Miyashita, ``{Noise
  covariance incorporated {MEG}-{MUSIC} algorithm: a method formultiple-dipole
  estimation tolerant of the influence of background brainactivity},''
  \emph{Biomedical Engineering, IEEE Transactions on}, vol.~44, no.~9, pp.
  839--847, 1997.

\bibitem{sekihara99a}
------, ``{{MEG} spatio-temporal analysis using a covariance matrix calculated
  from nonaveraged multiple-epoch data},'' \emph{Biomedical Engineering, IEEE
  Transactions on}, vol.~46, no.~5, pp. 515--521, 1999.

\bibitem{vantrees02a}
H.~L.~V. Trees, \emph{Detection, Estimation, and Modulation Theory Part IV:
  Optimum Array Processing}.\hskip 1em plus 0.5em minus 0.4em\relax new York:
  John wiley and Sons, Inc., 2002.

\bibitem{kailath-wax}
M.~Wax and T.~Kailath, ``Detection of signals by information theoretic
  criteria,'' \emph{IEEE Trans. Acoust. Speech Signal Process.}, vol.~33,
  no.~2, pp. 387--392, 1985.

\bibitem{zhao86a}
L.~C. Zhao, P.~R. Krishnaiah, and Z.~D. Bai, ``On detection of the number of
  signals in presence of white noise,'' \emph{J. Multivariate Anal.}, vol.~20,
  no.~1, pp. 1--25, 1986.

\bibitem{liavas01a}
A.~P. Liavas and P.~A. Regalia, ``On the behavior of information theoretic
  criteria for model order selection,'' \emph{IEEE Trans. Signal Process.},
  vol.~49, no.~8, pp. 1689--1695, August 2001.

\bibitem{raj08a}
R.~Nadakuditi and A.~Edelman, ``Sample eigenvalue based detection of
  high-dimensional signals in white noise using relatively few samples,''
  \emph{Signal Processing, IEEE Transactions on}, vol.~56, no.~7, pp.
  2625--2638, July 2008.

\bibitem{johnstone08a}
I.~M. Johnstone, ``Multivariate analysis and {J}acobi ensembles: largest
  eigenvalue, {T}racy--{W}idom limits and rates of convergence,'' \emph{Annals
  of {S}tatistics}, vol.~36, no.~6, pp. 2638--2716, 2008.

\bibitem{wishart28a}
J.~Wishart, ``The generalized product moment distribution in samples from a
  normal multivariate population,'' \emph{Biometrika}, vol. 20 A, pp. 32--52,
  1928.

\bibitem{muirhead82a}
R.~J. Muirhead, \emph{Aspects of multivariate statistical theory}.\hskip 1em
  plus 0.5em minus 0.4em\relax New York: John Wiley \& Sons Inc., 1982, wiley
  Series in Probability and Mathematical Statistics.

\bibitem{anderson03a}
T.~W. Anderson, \emph{An introduction to multivariate statistical analysis},
  3rd~ed., ser. Wiley Series in Probability and Statistics.\hskip 1em plus
  0.5em minus 0.4em\relax Hoboken, NJ: Wiley-Interscience [John Wiley \& Sons],
  2003.

\bibitem{silverstein85a}
J.~W. Silverstein, ``The limiting eigenvalue distribution of a multivariate {F}
  matrix,'' \emph{SIAM Journal on Math. Anal.}, vol.~16, no.~3, pp. 641--646,
  1985.

\bibitem{zhao86b}
L.~C. Zhao, P.~R. Krishnaiah, and Z.~D. Bai, ``On detection of the number of
  signals when the noise covariance matrix is arbitrary,'' \emph{J.
  Multivariate Anal.}, vol.~20, no.~1, pp. 26--49, 1986.

\bibitem{zhu91a}
Z.~Zhu, S.~Haykin, and X.~Huang, ``Estimating the number of signals using
  reference noise samples,'' \emph{IEEE Trans. on Aero. and Elec. Systems},
  vol.~27, no.~3, pp. 575--579, May 1991.

\bibitem{stoica97a}
P.~Stoica and M.~Cedervall, ``Detection tests for array procesing in unknown
  correlated noise fields,'' \emph{IEEE Trans. Signal Process.}, vol.~45, pp.
  2351--2362, September 1997.

\bibitem{xu94a}
G.~Xu, R.~H. Roy, and T.~Kailath, ``Detection of number of sources via
  exploitation of centro-symmetry property,'' \emph{IEEE Trans. Signal
  Process.}, vol. SP-42, pp. 102--112, January 1994.

\bibitem{larocque02a}
J.-R. Larocque, J.~P. Reilly, and W.~Ng, ``Particle filters for tracking and
  unknown number of sources,'' \emph{IEEE Trans. of Signal Processing},
  vol.~50, no.~12, pp. 2926--2937, December 2002.

\bibitem{silverstein:book}
Z.~D. Bai and J.~W. Silverstein, \emph{Spectral Analysis of Large Dimensional
  Random Matrices}.\hskip 1em plus 0.5em minus 0.4em\relax Beijing: Science
  Press, 2006.

\bibitem{marcenko67a}
V.~A. Mar{\v{c}}enko and L.~A. Pastur, ``Distribution of eigenvalues in certain
  sets of random matrices,'' \emph{Mat. Sb. (N.S.)}, vol. 72 (114), pp.
  507--536, 1967.

\bibitem{BaikBP04}
J.~Baik, G.~Ben~Arous, and S.~P{\'e}ch{\'e}, ``Phase transition of the largest
  eigenvalue for nonnull complex sample covariance matrices,'' \emph{Ann.
  Probab.}, vol.~33, no.~5, pp. 1643--1697, 2005.

\bibitem{BaikS06}
J.~Baik and J.~W. Silverstein, ``Eigenvalues of large sample covariance
  matrices of spiked population models,'' \emph{Journal of Multivariate
  Analysis}, no.~6, pp. 1382--1408, 2006.

\bibitem{paul07a}
D.~Paul, ``{Asymptotics of sample eigenstructure for a large dimensional spiked
  covariance model},'' \emph{Statistica Sinica}, vol.~17, no.~4, pp.
  1617--1642, 2007.

\bibitem{anandkumar09a}
A.~Anandkumar, L.~Tong, and A.~Swami, ``Detection of {G}auss-{M}arkov {R}andom
  {F}ields with nearest-neighbor dependency,'' \emph{Information Theory, IEEE
  Transactions on}, vol.~55, no.~2, pp. 816--827, Feb. 2009.

\bibitem{sung06a}
Y.~Sung, L.~Tong, and H.~Poor, ``{Neyman--Pearson Detection of Gauss--Markov
  Signals in Noise: Closed-Form Error Exponent and Properties},''
  \emph{Information Theory, IEEE Transactions on}, vol.~52, no.~4, pp.
  1354--1365, 2006.

\bibitem{rue2005gmr}
H.~Rue and L.~Held, \emph{{Gaussian Markov Random Fields: Theory and
  Applications}}.\hskip 1em plus 0.5em minus 0.4em\relax Chapman \& Hall/CRC,
  2005.

\bibitem{vempala2007sal}
S.~Vempala, ``{Spectral Algorithms for Learning and Clustering},''
  \emph{LECTURE NOTES IN COMPUTER SCIENCE}, vol. 4539, p.~3, 2007.

\bibitem{mangasarian2006mps}
O.~Mangasarian and E.~Wild, ``{Multisurface Proximal Support Vector Machine
  Classification via Generalized Eigenvalues},'' \emph{IEEE TRANSACTIONS ON
  PATTERN ANALYSIS AND MACHINE INTELLIGENCE}, pp. 69--74, 2006.

\bibitem{guarracino2007cmb}
M.~Guarracino, C.~Cifarelli, O.~Seref, and P.~Pardalos, ``{A classification
  method based on generalized eigenvalue problems},'' \emph{Optimization
  Methods and Software}, vol.~22, no.~1, pp. 73--81, 2007.

\bibitem{TracyW94}
C.~Tracy and H.~Widom, ``Level-spacing distribution and {Airy} kernel,''
  \emph{Communications in Mathematical Physics}, vol. 159, pp. 151--174, 1994.

\bibitem{TracyW96}
------, ``On orthogonal and symplectic matrix ensembles,'' \emph{Communications
  in Mathematical Physics}, vol. 177, pp. 727--754, 1996.

\bibitem{johnstone07a}
I.~M. Johnstone, ``High dimensional statistical inference and random
  matrices,'' in \emph{International {C}ongress of {M}athematicians. {V}ol.
  {I}}.\hskip 1em plus 0.5em minus 0.4em\relax Eur. Math. Soc., Z\"urich, 2007,
  pp. 307--333.

\bibitem{ElKaroui07a}
N.~El~Karoui, ``Tracy-{W}idom limit for the largest eigenvalue of a large class
  of complex sample covariance matrices,'' \emph{Ann. Probab.}, vol.~35, no.~2,
  pp. 663--714, 2007.

\bibitem{ma08a}
Z.~Ma, ``Accuracy of the {T}racy--{W}idom limit for the largest eigenvalue in
  white {W}ishart matrices,'' 2008, {\tt http://arxiv.org/abs/0810.1329}.

\bibitem{silverstein95b}
J.~W. Silverstein, ``Strong convergence of the empirical distribution of
  eigenvalues of large dimensional random matrices,'' \emph{J. of Multivariate
  Anal.}, vol. 55(2), pp. 331--339, 1995.

\bibitem{silverstein95a}
J.~W. Silverstein and S.-I. Choi, ``Analysis of the limiting spectral
  distribution of large-dimensional random matrices,'' \emph{J. Multivariate
  Anal.}, vol.~54, no.~2, pp. 295--309, 1995.

\bibitem{yinbai88a}
Y.~Q. Yin, Z.~D. Bai, and P.~R. Krishnaiah, ``On the limit of the largest
  eigenvalue of the large-dimensional sample covariance matrix,'' \emph{Probab.
  Theory Related Fields}, vol.~78, no.~4, pp. 509--521, 1988.

\bibitem{bai93c}
Z.~D. Bai and Y.~Q. Yin, ``Limit of the smallest eigenvalue of a
  large-dimensional sample covariance matrix,'' \emph{Ann. Probab.}, vol.~21,
  no.~3, pp. 1275--1294, 1993.

\bibitem{bai98z}
Z.~D. Bai and J.~W. Silverstein, ``No eigenvalues outside the support of the
  limiting spectral distribution of large dimensional sample covariance
  matrices,'' \emph{Ann. Probab.}, vol. 26, No. 1, pp. 316--345, 1998.

\bibitem{bai99z}
------, ``Exact separation of eigenvalues of large dimensional sample
  covariance matrices,'' \emph{Ann. Probab.}, vol. 27, No. 3, pp. 1536--1555,
  1999.

\bibitem{billingsley95a}
P.~Billingsley, \emph{Probability and measure}, 3rd~ed., ser. Wiley Series in
  Probability and Mathematical Statistics.\hskip 1em plus 0.5em minus
  0.4em\relax New York: John Wiley \& Sons Inc., 1995, a Wiley-Interscience
  Publication.

\end{thebibliography}
